\documentclass[10pt]{article}

\usepackage[utf8]{inputenc}
\usepackage[T1]{fontenc}

\usepackage{epsf}
\usepackage{amsmath}

\allowdisplaybreaks

\usepackage[showframe=false]{geometry}
\usepackage{changepage}

\usepackage{epsfig}
\usepackage{amssymb}

\usepackage{amsthm}
\usepackage{setspace}
\usepackage{cite}
\usepackage{mcite}

\usepackage{algorithmic}  
\usepackage{algorithm}

\usepackage{shadow}
\usepackage{fancybox}
\usepackage{fancyhdr}

\usepackage{color}
\usepackage[usenames,dvipsnames,svgnames,table]{xcolor}
\newcommand{\bl}[1]{\textcolor{blue}{#1}}

\definecolor{mypurple}{rgb}{.4,.0,.5}
\newcommand{\prp}[1]{\textcolor{mypurple}{#1}}

\usepackage[hyphens]{url}

\usepackage[colorlinks=true,
            linkcolor=black,
            urlcolor=blue,
            citecolor=purple]{hyperref}

\usepackage{breakurl}

\def\y{{\bf y}}

\def\x{{\bf x}}

\def\x{{\mathbf x}}

\def\u{{\bf u}}

\def\x{{\bf x}}
\def\y{{\bf y}}

\def\q{{\bf q}}
\def\m{{\bf m}}

\def\c{{\bf c}}

\def\h{{\bf h}}

\def\cH{{\mathcal H}}

\def\be{\begin{equation}}
\def\ee{\end{equation}}
\def\ba{\left[\begin{array}}
\def\ea{\end{array}\right]}

\def\t{{\bf t}}

\def\u{{\bf u}}

\def\x{{\bf x}}
\def\y{{\bf y}}

\def\q{{\bf q}}

\def\c{{\bf c}}

\def\p{{\bf p}}

\def\1{{\bf 1}}

\def\0{{\bf 0}}

\def\calX{{\cal X}}
\def\calY{{\cal Y}}

\def\mR{{\mathbb R}}
\def\mN{{\mathbb N}}
\def\mE{{\mathbb E}}
\def\mS{{\mathbb S}}
\def\mB{{\mathbb B}}

\def\lp{\left (}
\def\rp{\right )}

\def\y{{\bf y}}

\def\x{{\bf x}}

\def\x{{\mathbf x}}

\def\u{{\bf u}}

\def\x{{\bf x}}
\def\y{{\bf y}}

\def\q{{\bf q}}

\def\c{{\bf c}}

\def\h{{\bf h}}

\def\cH{{\cal H}}

\def\be{\begin{equation}}
\def\ee{\end{equation}}
\def\ba{\left[\begin{array}}
\def\ea{\end{array}\right]}

\def\t{{\bf t}}

\def\u{{\bf u}}

\def\x{{\bf x}}
\def\y{{\bf y}}

\def\q{{\bf q}}

\def\c{{\bf c}}

\def\p{{\bf p}}

\def\({\left (}
\def\){\right )}

\def\1{{\bf 1}}
\def\m{{\bf m}}
\def\q{{\bf q}}

\def\0{{\bf 0}}

\def\cX{{\mathcal X}}
\def\cY{{\mathcal Y}}

\usepackage{xcolor}
\usepackage{color}

\definecolor{darkgreen}{rgb}{0, 0.4,0}

\definecolor{purplebrown}{rgb}{0.5,0.1,0.6}

\definecolor{ultclupcol}{rgb}{0.1,0.5,0.5}

\definecolor{mytrycolor}{rgb}{0.5,0.7,0.2}

\definecolor{ultclupcola}{rgb}{.5,0,.5}

\definecolor{shadebrown}{rgb}{0.1,0.1,0.9}
\definecolor{lightblue}{rgb}{0.2,0,1}

\usepackage{fancybox}
\usepackage{graphicx}
\usepackage{epstopdf}
\usepackage{epsfig}
\usepackage{wrapfig}
\usepackage{subfigure}

\usepackage{xcolor}
\usepackage{tcolorbox}

\newtcbox{\xmybox}{on line,
arc=7pt,
before upper={\rule[-3pt]{0pt}{10pt}},boxrule=0pt,
boxsep=0pt,left=6pt,right=6pt,top=0pt,bottom=0pt,enhanced, coltext=blue, colback=white!10!yellow}

\newtcbox{\xmyboxa}{on line,
arc=7pt,
before upper={\rule[-3pt]{0pt}{10pt}},boxrule=0pt,
boxsep=0pt,left=6pt,right=6pt,top=0pt,bottom=0pt,enhanced, colback=white!10!yellow}

\newtcbox{\xmyboxb}{on line,
arc=7pt,
before upper={\rule[-3pt]{0pt}{10pt}},boxrule=1pt,colframe=darkgreen!100!blue,
boxsep=0pt,left=6pt,right=6pt,top=0pt,bottom=0pt,enhanced, colback=white!10!yellow}

\newtcbox{\xmyboxc}{on line,
arc=7pt,
before upper={\rule[-3pt]{0pt}{10pt}},boxrule=.7pt,colframe=blue!100!blue,
boxsep=0pt,left=6pt,right=6pt,top=0pt,bottom=0pt,enhanced, coltext=blue, colback=white!10!yellow}

\newtcbox{\xmytboxa}{on line,
arc=7pt,
before upper={\rule[-3pt]{0pt}{10pt}},boxrule=.0pt,colframe=pink!50!yellow,
boxsep=0pt,left=6pt,right=6pt,top=0pt,bottom=0pt,enhanced, coltext=white, colback=blue!40!red}

\newtcbox{\xmytboxb}{on line,
arc=7pt,
before upper={\rule[-3pt]{0pt}{10pt}},boxrule=.0pt,colframe=pink!50!yellow,
boxsep=0pt,left=6pt,right=6pt,top=0pt,bottom=0pt,enhanced, coltext=white, colback=white!40!green}

\makeatletter
\newcommand\subsubsubsection{\@startsection{paragraph}{4}{\z@}{-2.5ex\@plus -1ex \@minus -.25ex}{1.25ex \@plus .25ex}{\normalfont\normalsize\bfseries}}
\newcommand\subsubsubsubsection{\@startsection{subparagraph}{5}{\z@}{-2.5ex\@plus -1ex \@minus -.25ex}{1.25ex \@plus .25ex}{\normalfont\normalsize\bfseries}}
\makeatother

\newtheorem{theorem}{Theorem}

\begin{document}

\begin{singlespace}

\title {CLuP practically achieves $\sim 1.77$ positive and $\sim 0.33$ negative Hopfield model ground state free energy 
}
\author{
\textsc{Mihailo Stojnic
\footnote{e-mail: {\tt flatoyer@gmail.com}} }}
\date{}
\maketitle

\centerline{{\bf Abstract}} \vspace*{0.1in}

We study algorithmic aspects of finding $n$-dimensional  \emph{positive} and \emph{negative} Hopfield ($\pm$Hop) model ground state free energies. This corresponds  to classical maximization of random positive/negative semi-definite quadratic forms over binary $\left \{\pm \frac{1}{\sqrt{n}} \right \}^n$ vectors. The key algorithmic question is whether these problems can be computationally efficiently approximated within a factor $\approx 1$.  Following the introduction and success of \emph{Controlled Loosening-up} (CLuP-SK) algorithms in finding near  ground state energies of closely related Sherrington-Kirkpatrick (SK) models \cite{Stojnicclupsk25}, we here propose a CLuP$\pm$Hop counterparts for $\pm$Hop models.  Fully lifted random duality theory (fl RDT) \cite{Stojnicflrdt23} is utilized to characterize  CLuP$\pm$Hop \emph{typical} dynamics. An excellent agreement between practical performance  and theoretical predictions is observed. In particular, for $n$  as small as few thousands CLuP$\pm$Hop achieve $\sim 1.77$ and  $\sim 0.33$ as the ground state free energies of the positive and negative Hopfield models. At the same time we obtain on the 6th level of lifting (6-spl RDT) corresponding theoretical thermodynamic ($n\rightarrow\infty$) limits $\approx 1.7784$ and $\approx 0.3281$. This positions determining Hopfield models near ground state energies as \emph{typically} easy problems. Moreover, the very same 6th lifting level evaluations allow to uncover a fundamental intrinsic difference between two models: $+$Hop's near optimal configurations are \emph{typically close} to each other whereas the $-$Hop's are \emph{typically far away}.

\vspace*{0.25in} \noindent {\bf Index Terms: Hopfield models; CLuP algorithm; Fully lifted random duality theory}.

\end{singlespace}

\section{Introduction}
\label{sec:back}

Taking an integer $n\in\mN$ and matrix $A\in\mR^{n\times n}$  we are interested in the  following quadratic form optimization
\begin{eqnarray}\label{eq:inteq1}
 \max_{\x\in\mB^n} \x^TA\x,
\end{eqnarray}
where $\mB^n$ are the vertices of the following  $n$-dimensional ``binary'' cube
\begin{eqnarray}\label{eq:inteq2}
\mB^n \triangleq \left \{ \x| \x\in\mR^n, \x_i^2=\frac{1}{n}\right \}.
\end{eqnarray}
Being one of the  foundational problems of the classical NP complexity theory, (\ref{eq:inteq1}) is hard to approximate within a $\log(n)^{const.}$ factor \cite{AroraBKSH05} (for further  considerations including relaxations of integer constraints, see, e.g., \cite{CharikarW04,Meg01}). While NP concepts are a useful generic computational complexity guideline, their \emph{worst case} nature might leave them somewhat inconvenient when it comes to  proper addressing of \emph{typical} solvability. There indeed may be few bad instances of (\ref{eq:inteq1}) that are unsolvable in polynomial time, but there is no guarantee that the same scenario repeats itself for a sizeable portion (or even less so for a majority) of all instances. In other words, hardness of the worst case does not say much about \emph{typical} solvability. For example, the subset of instance of (\ref{eq:inteq1}) where $A$ is  positive semi-definite ($A\succeq 0$) is approximable within constant $\frac{2}{\pi}$ factor \cite{Nest97}. Things even more dramatically change as one gets to random mediums.  For (indefinite) $A$ comprised of standard normals (\ref{eq:inteq1}) becomes
famous Sherrington-Kirkpatrick (SK) model \cite{SheKir72} where simple spectral methods ensure (probabilistic) approximability within a constant factor thereby strongly improving over the above mentioned $\log(n)^{const.}$ inaapproximability.  Moreover, recent studies \cite{Montanari19,Stojnicclupsk25} move things much further and effectively uncover determining SK-models ground state energies as easy problems.

Following into the footsteps of  \cite{Stojnicclupsk25} we here consider \emph{positive} and \emph{negative} Hopfield ($\pm$Hop) models whose ground state energies are obtained for positive and negative semi-definite forms $A$. As SK  models, Hopfiled models have comparably long history \cite{Hop82,PasFig78,Hebb49,PasShchTir94,ShchTir93,BarGenGueTan10,BarGenGueTan12,Tal98,StojnicMoreSophHopBnds10,BovGay98,Zhao11,Talbook11a,Talbook11b,Stojnicnflgscompyx23,Stojnichopflrdt23}. Starting out as concepts of cognitive learning \cite{Hop82,Hebb49} they quickly became integral parts of studies in many different areas including physics and mathematics  \cite{Hop82,PasFig78,PasShchTir94,ShchTir93,KroHop16,BarGenGueTan10,BarGenGueTan12,Tal98,StojnicMoreSophHopBnds10,BovGay98,Zhao11,Talbook11a,Talbook11b,Stojnichopflrdt23}, statistics, signal processing, algorithmic computer science
\cite{CharikarW04,Meg01,Nest97,StojnicISIT2010binary,Stojnicclupint19,Stojnicclupspreg20,BayMon10,BayMon10lasso,StojnicCSetam09,StojnicGenLasso10},
and most notably neural networks
\cite{AGS87,AmiGutSom87,KroHop16,Rametal21,Stojnichebhop24,Newman88,MPRV87,FST00,Louk94a,Louk94,Louk97,Hop82,PasFig78,Hebb49}. Almost all of these applications directly relate to various forms of $\pm$Hop models ground state energies and impose finding ways to efficiently compute them as a rather pressing need.

Motivated by the success of the \emph{Controlled Loosening-up} (CLuP) algorithms in handling so-called \emph{planted} models in compressed sensing, statistical regression, and MIMO ML detections \cite{Stojnicclupint19,Stojnicclupspreg20},  \cite{Stojnicclupsk25} proposed  CLuP-SK counterpart for handling \emph{non-planted} SK models.  Building up on \cite{Stojnicclupsk25}, we here propose a CLuP$\pm$Hop alternative for  $\pm$Hop models. Somewhat paradoxically, keeping in mind that \cite{Stojnicclupint19,Stojnicclupspreg20}'s applications are actually of the (negative) Hopfield type, we in a way reconnect  CLuP back to its origin. However,  this reconnection is now way more challenging as  \emph{non-planted} scenarios are algorithmically much harder. Before we switch to technical intricacies that shed more light on this, we briefly introduce  needed mathematical basics and  overview some of the most relevant prior results.

\section{Hopfield ($\pm$Hop) models --- basics, prior work, and contributions}
\label{sec:skmodel}

Form in (\ref{eq:inteq1}) is a special case of a  general statistical mechanics concept. For matrix of quenched interactions $G$, consider the so-called Hamiltonian
\begin{equation}
\cH(G)=\sum_{i=1}^{n} \sum_{j=1}^{m} G_{ij}\x_i\y_j = \y^TG\x,\label{eq:a0ham1}
\end{equation}
and for the \emph{inverse temperature} parameter $\beta>0$ and the $m$-dimensional unit norm sphere $\mS^m$ associate to it the following  partition function
\begin{equation}
Z(\beta,G)=\sum_{\x\in\mB^n} \lp \sum_{\y\in\mS^m}  e^{\beta\cH(G)} \rp^s .\label{eq:a0partfun}
\end{equation}
Throughout the paper we are interested in random $\cH(G)$ with elements of $G$ being independent standard normals. We consider so-called linear/proportional regimes with $\alpha = \lim_{n\rightarrow\infty} \frac{m}{n}$ remaining constant as $n$ grows. One then has for
  the thermodynamic limit ($n\rightarrow\infty$) (average) free energy
\begin{equation}
f_{sq}(\beta)=\lim_{n\rightarrow\infty}\frac{\mE_G\log{(Z(\beta,G)})}{\beta \sqrt{n}}
=\lim_{n\rightarrow\infty} \frac{\mE_G\log \lp  \sum_{\x\in\mB^n} \lp \sum_{\y\in\mS^m}  e^{\beta\cH(G)} \rp^s  \rp    }{\beta \sqrt{n}},\label{eq:a0logpartfunsqrt}
\end{equation}
with $\mE_G$ standing for the expectation with respect to $G$ (as a general notational convention,  subscripts of $\mE$ denote the randomness for which the expectation is evaluated). Equipped with all the technicalities we can now introduce the ground state free energy as the following zero-temperature,
$\beta\rightarrow\infty$, thermodynamic limit
\begin{eqnarray}
f_{sq}(\infty)  & \triangleq  & \lim_{\beta\rightarrow\infty}f_{sq}(\beta)  =
\lim_{\beta,n\rightarrow\infty}\frac{\mE_G\log{(Z(\beta,G)})}{\beta \sqrt{n}}=
 \lim_{n\rightarrow\infty}\frac{\mE_G \max_{\x\in\mB^n} s \max_{\y\in\mS^m} \cH(G)}{\sqrt{n}}
  \nonumber \\
&  =  & \lim_{n\rightarrow\infty}\frac{\mE_G \max_{\x\in\mB^n} s  \max_{\y\in\mS^m} \y^TG\x}{ \sqrt{n}}.
  \label{eq:a0limlogpartfunsqrt}
\end{eqnarray}
In particular, we have  for $s=1$
\begin{eqnarray}
f_{sq}^+(\infty)  &= &
 \lim_{n\rightarrow\infty}\frac{\mE_G \max_{\x\in\mB^n}  \max_{\y\in\mS^n}  \cH(G)}{\sqrt{n}}  = \lim_{n\rightarrow\infty}\frac{\mE_G \max_{\x\in\mB^n} \max_{\x\in\mS^m}\y^TG\x}{\sqrt{n}}
 \nonumber \\
& = & \lim_{n\rightarrow\infty}\frac{\mE_G \max_{\x\in\mB^n} \sqrt{\x^TG^T G\x}} {\sqrt{n}},
  \label{eq:a0limlogpartfunsqrta0}
\end{eqnarray}
and for $s=-1$
\begin{eqnarray}
f_{sq}^- (\infty)  &= &
 \lim_{n\rightarrow\infty}\frac{\mE_G \max_{\x\in\mB^n} - \max_{\y\in\mS^n}  \cH(G)}{\sqrt{n}}  = \lim_{n\rightarrow\infty}\frac{\mE_G \max_{\x\in\mB^n} - \max_{\x\in\mS^m}  \y^TG\x}{\sqrt{n}}
 \nonumber \\
& = & \lim_{n\rightarrow\infty}\frac{\mE_G \max_{\x\in\mB^n} - \sqrt{\x^TG^T G\x}} {\sqrt{n}}
 = - \lim_{n\rightarrow\infty}\frac{\mE_G \min_{\x\in\mB^n}  \sqrt{\x^TG^T G\x}} {\sqrt{n}}.
  \label{eq:a0limlogpartfunsqrta1}
\end{eqnarray}
It is now not that difficult to see that (\ref{eq:a0limlogpartfunsqrta0}) is computationally equivalent to (\ref{eq:inteq1}) with $A=G^TG$ being positive semi-definite. Analogously, (\ref{eq:a0limlogpartfunsqrta1}) is computationally equivalent to (\ref{eq:inteq1}) with $A=-G^TG$ being negative semi-definite. As precisely such forms of $A$ correspond to \emph{positive} and \emph{negative} Hopfield ($\pm$Hop) models,  the ground state regime and the associated free energies from  (\ref{eq:a0limlogpartfunsqrta0}) and  (\ref{eq:a0limlogpartfunsqrta1}) are our main focus. However,  we find working with the general free energy form from (\ref{eq:a0logpartfunsqrt}) as more convenient. Eventually the ground state $\pm$Hop's behaviors are deduced as  special ($\beta\rightarrow\infty$)  cases of (\ref{eq:a0logpartfunsqrt}).

\subsection{Prior work}
\label{sec:skmodelkey}

\noindent \underline{\emph{\textbf{Theoretical aspects:}}}  Two spin-glass models -- the  Edwards-Anderson (EW) nearest neighbor  \cite{EdwAnd75} and Sherrington-Kirkpatrick (SK) long range antipode \cite{SheKir72} --  proposed in the mid seventies of the last century to a large degree shaped up studies in statistical mechanics and many associated computational and engineering fields over the last 50 years. Analysis of the models via  \cite{EdwAnd75}'s replica method uncovered incredibly rich intrinsic structures.  While it was clear that something big is around the corner, a few unpleasant logical inconsistencies appeared as well. First \cite{SheKir72}  observed the so-called negative entropy crisis and then extensive followup simulations  \cite{SheKir78} further uncovered a non-negligible disagreement between the algorithmically computed and theoretically predicted ground state free energies. Attempting to resolve the mismatches, Parisi in a  breakthrough discovery \cite{Par79,Parisi80,Par80} proposed replacing \cite{EdwAnd75}'s \emph{replica-symmetry} (RS) ansatz with a more refined \emph{replica symmetry breaking} (RSB) one.  He demonstrated that already after applying RSB  with two steps of breaking, the negative entropy crisis almost completely disappears. Moreover, recognizing that RSB corrections are more pronounced as $\beta$ grows he also found that two steps of RSB lower the RS ground state energy prediction $\approx 0.7979$ to $\approx 0.7636$ thereby achieving a stronger  agreement with $\sim 0.75\pm 0.01$ estimate from \cite{SheKir78}. Belief in the Parisi RSB quickly strengthened  and soon after the appearance of \cite{Par79,Parisi80,Par80,Par83} it became an irreplaceable tool in utilization of replica methods. While horizons of applications widened, rigorous analytical justifications lagged behind. About 25 years later,  Guerra  \cite{Guerra03} and Talagrand  \cite{Tal06} proved that Parisi RSB characterization of the SK model is  correct. Panchenko  \cite{Pan10,Pan10a,Pan13,Pan13a} reproved these results and additionally established the validity of the \emph{ultrametricity} \cite{Par83} (for further studies in these and related directions, see, e.g., \cite{AuffC15,AuffCZ20,JagTob16} and references therein).

While  all of the above  settled the theoretical aspects of the SK models, extensions to  more complex models have been noticeably absent. In fact, not only  are they not straightforward, but different methodologies are actually needed. Hopfield models are excellent examples where this can be seen. After their invention within cognitive learning context \cite{Hop82,Hebb49},  they have been connected to spin glass models \cite{PasFig77} and extensively studied within neural networks (NN) communities. Introductory  considerations \cite{Hop82,MPRV87} were followed by \cite{CriAmiGut86,SteKuh94,AmiGutSom87} where free energies were studied for general $\beta$'s via replica methods (some of the problems studied within NN context are actually even more general than the form in (\ref{eq:inteq1})). Early mathematical treatments \cite{PasShchTir94,ShchTir93,BarGenGueTan10,Tal98,Zhao11,BarGenGue11bip} were typically restricted to favorable dimensional and high-temperature ($T$) regimes (low $\beta=\frac{1}{T}$), where the so-called replica symmetric behavior is present. Outside such regimes, the problems are hard and precise studies were usually replaced by bounding efforts \cite{Newman88,Louk94,Louk94a,Louk97,Tal98}. Completely opposite  of  the high-temperature regime, the ground state  assumes zero temperature. Consequently, related available results are even scarcer. Relying on the random duality (RDT) theory, \cite{StojnicHopBnds10}  upper bounded the free energy and showed that such upper bounds precisely match the RS based predictions \cite{DeanRit01}.  Utilizing a \emph{lifted} RDT variant, \cite{StojnicMoreSophHopBnds10}  lowered bounds of \cite{StojnicHopBnds10} and showed that RS predictions are \emph{not} tight implying that symmetry has to be broken. Moreover,  for $+$Hop \cite{StojnicMoreSophHopBnds10} produced analogous lower bounds, effectively giving the ranges $f_{sq}^+(\infty)\in[1.7632,1.7832]$ and $f_{sq}^-(\infty)\leq -0.3202$. Finally,  fully lifted RDT \cite{Stojnicflrdt23} allowed  precise evaluation  of $\pm$Hop ground state free energies. In particular, on the 3rd level of lifting,  \cite{Stojnichopflrdt23} obtained  $f_{sq}^{+,3}(\infty)\approx 1.7789$ and $f_{sq}^{-,3}(\infty)\approx -0.3279$ (in Section \ref{sec:properties} of this paper, we obtain on the 6th level of lifting $f_{sq}^{+,6}(\infty)\approx 1.77842$ and $f_{sq}^{-,6}(\infty)\approx -0.32807$).

\noindent \underline{\emph{\textbf{Algorithmic aspects:}}}  After determining $\pm$Hop models  theoretical limits the natural  algorithmic followup  question is whether spin configurations that achieve such limits can be efficiently  found. As mentioned earlier, within classical NP complexity theory for an indefinite $A$ the optimum in (\ref{eq:inteq1}) is hard to approximate within a $\log(n)^{const.}$ factor \cite{AroraBKSH05}. On the other hand, for any $A\succeq 0$,  semi-definite programming (SDP) relaxations ensure constant factor  $\frac{2}{\pi}$  approximation  \cite{Nest97}. As  NP concepts are \emph{worst case} based they rarely give a proper assessment regarding  \emph{typical} solvability. A more faithful representation in that regard is obtained after one switches to random mediums (this also happens to be a natural setup of $\pm$Hop models). In such context,  \cite{StojnicHopBnds10} obtained that a simple rounding leading eigenvector procedure gives $f_{sq}^+(\infty) = 2\sqrt{\frac{2}{\pi}}$. Combining this with a trivial leading eigenvalue bound  $f_{sq}^+(\infty) \leq 2$ gives $f_{sq}^+(\infty)   \in \left \{\sqrt{\frac{8}{\pi}},2 \right \}$. Somewhat paradoxically, one obtains the approximating factor of $\frac{2}{\pi}$ (it is actually  $\sqrt{\frac{2}{\pi}}$ since $f_{sq}^+(\infty)$ is the root of the optimum in (\ref{eq:inteq1}) for $A\succeq 0$). In other words simple spectral considerations from  \cite{StojnicHopBnds10}  give exactly the same approximating factor in random context that  \cite{Nest97} gives in non-random. As further hinted in \cite{StojnicHopBnds10}, the given range $f_{sq}(\infty)   \in \left \{\sqrt{\frac{8}{\pi}},2 \right \}$ is fairly pessimistic. In fact, another simple iterative procedure presented in  \cite{StojnicHopBnds10} gave similar guarantees but practically performed much better producing  $f_{sq}^+(\infty) \approx 1.70$ which is much closer to $f_{sq}^+(\infty)\approx 1.7784$. On the other hand, analogous procedure for $-$Hop model gave $f_{sq}^-(\infty) \approx -0.55$ which is far away from $f_{sq}^-(\infty)\approx -0.3281$. While all of the above put a strong effort in improving the constant factor approximation, it still left one wonder if the approximating factor can approach 1?  Or in other words, whether or not the Hopfield models exhibit  \emph{computational gaps}? If the approximative factor can be arbitrarily close to 1, then $\pm$Hop models have no computational gaps and are typically easy (for more on similar phenomena, where  problems hard within classical NP theory are actually typically easy, see, e.g., \cite{Gamar21,GamarSud14,GamarSud17,GamarSud17a} and references therein).

We should also mention a few groundbreaking recent results related to similar statistical mechanics models.  A continuous random energy model (CREM) was considered in  \cite{AbM19} and shown to be solvable in polynomial time. The problem/model is somewhat artificial but it is of great use as it is presumed sufficiently similar to SK when it comes to studying algorithmic properties.  Subag in \cite{Subag21} studied $p$-spin spherical models (also closely related to SK) and designed an efficient algorithm to solve them. Partially inspired by \cite{Subag21},  Montanari in \cite{Montanari19} provided analogous breakthrough for SK model itself. Differently from \cite{Subag21},  \cite{Montanari19} relied on  modifications of the message passing algorithms \cite{Kaba03,DonMalMon09,DonohoMM11}. Starting with famous approximate message passing (AMP)  \cite{DonMalMon09},  \cite{Montanari19}  considers its  an incremental variant, IAMP. and shows that it can compute arbitrarily closely SK ground state energy provided that the Parisi RSB parametric functional is monotonically increasing (or equivalently that the overlap gap property (OGP)  is absent \cite{Gamar21,GamarSud14,GamarSud17,GamarSud17a,AchlioptasCR11,HMMZ08,MMZ05}). Despite missing the formal proofs,  these properties are commonly believed to be true (for further extensions in these and closely related directions see also \cite{AlaouiMS22,AlaouiMS25,HuangS22,Subag24,Subag17,Subag17a}). It would be very interesting to see if the results from \cite{Montanari19} (and \cite{Subag21}) can be extended to Hopfield models of interest here.

Following the success of \emph{Control Loosening-up} (CLuP) algorithms in planted models in statistics and signal processing \cite{Stojnicclupint19,Stojnicclupspreg20}, \cite{Stojnicclupsk25} introduced a CLuP-SK algorithm tailored for non-planted SK models. Its practical performance was shown to closely match corresponding theoretical predictions. Moreover, already for the dimensions on the order of few thousands it achieves $\sim 76$ SK model ground state energy which remarkably closely approaches the thermodynamic $n\rightarrow\infty$ limit of $\approx 0.7632$.

\subsection{Our contributions}
\label{sec:cont}

Given that CluP-SK  \cite{Stojnicclupsk25}  provided a nice non-planted model complement for original planted models applications \cite{Stojnicclupint19,Stojnicclupspreg20}, it is natural to wonder if  the range of non-planted scenarios solvable by CLuP like algorithms can be further extended. Focusing on Hopfield models as another class of non-planted models, we propose  the following implementation particularly tailored for them
 \begin{eqnarray}\label{eq:algimpeq2}
 \hspace{-.55in} \mbox{\bl{\textbf{\emph{CLuP$\pm$Hop algorithm:}}}}  \hspace{.65in} \x^{(t+1)} & \rightarrow  &
\mbox{\textbf{gradbar}}\lp\bar{f}_{b,x}^{\pm} \lp \x;\bar{t}_{0x}^{(t)} \rp ;\x^{(t)},\bar{t}_{0x}^{(t)} \rp
 \nonumber \\
\bar{t}_{0x}^{(t+1)}  &  \rightarrow  & \bar{c}^{(t)}\bar{t}_{0x}^{(t)}.
\end{eqnarray}
For starting $\x^{(t)}$, procedure $\mbox{\textbf{gradbar}}\lp\bar{f}_{b,x}^+ \lp \x;\bar{t}_{0x}^{(t)} \rp ;\x^{(t)},\bar{t}_{0x}^{(t)} \rp$ applies gradient descent to  function
\begin{eqnarray}\label{eq:algimpeq3}
\bar{f}_{b,x}^+ \lp\x;\bar{t}_{0x}^{(t)}\rp = - \bar{t}_{0x}^{(t)} \|\x\|_2 - \log\lp - \lp \x^T  \lp 4 I  -  \frac{1}{n}  G^TG   \rp \x - \kappa \rp \rp
-\frac{1}{n}\sum_{i=1}^{n}  \log(1-n\x_i^2),
\end{eqnarray}
whereas
$\mbox{\textbf{gradbar}}\lp\bar{f}_{b,x}^- \lp \x;\bar{t}_{0x}^{(t)} \rp ;\x^{(t)},\bar{t}_{0x}^{(t)} \rp$ applies gradient descent to  function
\begin{eqnarray}\label{eq:algimpeq3a0}
\bar{f}_{b,x}^- \lp\x;\bar{t}_{0x}^{(t)}\rp = - \bar{t}_{0x}^{(t)} \|\x\|_2 - \log\lp - \lp \x^T  \lp 0 I  +  \frac{1}{n}  G^TG   \rp \x - \kappa \rp \rp
-\frac{1}{n}\sum_{i=1}^{n}  \log(1-n\x_i^2).
\end{eqnarray}
Clearly, $\bar{f}_{b,x}^+ \lp\x;\bar{t}_{0x}^{(t)}\rp$ relates to $+$Hop and $\bar{f}_{b,x}^- \lp\x;\bar{t}_{0x}^{(t)}\rp$ to $-$Hop model. Also, we found that the above procedure is not overly sensitive to changes of free parameter $\kappa$. Choosing $\kappa=0.855$ for $+$Hop and $\kappa=0.115$ for $-$Hop  worked well in all our numerical experiments. Other parameters are also fairly flexible.  For example, $t_{0x}^{(0)}=0.1$ for $+$Hop and  $t_{0x}^{(0)}=0.001$ for $-$Hop together with $c^{(t)}=1.1$ are good starting options that can be adapted as $n$ changes. $\x^{(0)}$ is practically any $\x$ admissible under $\log$s. Taking a random point from  $\left \{ -\frac{1}{\sqrt{n}},\frac{1}{\sqrt{n}}\right \}^n$ and then scaling down by two until a feasible point is reached is a possible choice. Various stopping criteria can be used as well. One option is to take $\bar{t}_{0x}^{(t+1)} \geq  10^4$.

Let $\hat{\x}$ be the output of CLuP$\pm$Hop. Set
\begin{eqnarray}\label{eq:algimpeq4}
\hat{\xi} \triangleq \frac{\mE_G \sqrt{\mbox{sign}(\x)^TG^TG\mbox{sign}(\x)}  }{n}.
\end{eqnarray}
Table \ref{tab:tab2}  shows the obtained results for $\hat{\xi}$ for both positive and negative Hopfield models. As can be seen even for $n$ as small as  a few thousands one fairly closely approaches the thermodynamic $n\rightarrow\infty$ limits $\approx 1.7784$ and $\approx 0.3281$. It should also be noted that the results in the table are obtained for plain gradient descent without any restarts. In other words, they are obtained without any advanced modifications (say, stochastic gradient descent, multiple restarts with carefully chosen restarting points, interactive  $t_{0x}^{(0)}=0.1$ and $c^{(t)}$ retuning, and so on). The effects of such modifications are likely to fade away as $n\rightarrow \infty$, but for finite $n$  they could be  beneficial. Since the results shown in Table \ref{tab:tab2} are already significantly surpassing the best of the expectations, we found no point in overwhelming presentation with extensive discussions regarding additional modifications.

\begin{table}[h]
\caption{Performance of CLuP$\pm$Hop algorithm; \textbf{\bl{simulated}/theory} }\vspace{.1in}
\centering
\def\arraystretch{1.2}
\begin{tabular}{||l||c||c||c||c|| }\hline\hline
 \hspace{-0in}$n$                                             & $2000$    & $4000$ & $8000$ &  $\infty$ (\textbf{theory}) \\ \hline\hline
$\hat{\xi}$ ($+$Hop)                                         & \bl{$\mathbf{1.7704}$}  & \bl{$\mathbf{1.7721}$}  & \bl{$\mathbf{1.7735}$} & $\mathbf{1.7784}$  \\ \hline\hline
$\hat{\xi}$ ($-$Hop)                                         & \bl{$\mathbf{0.3355}$}  & \bl{$\mathbf{0.3340}$}  & \bl{$\mathbf{0.3330}$} & $\mathbf{0.3281}$  \\ \hline\hline
\end{tabular}
\label{tab:tab2}
\end{table}

To study properties of the above algorithm, we  fix $r_x$ ($0<r_x\leq 1$) and $\bar{r}_x<0$, and introduce the following
\begin{eqnarray}\label{eq:clupskeq1}
\hspace{-1.5in} \mbox{\bl{\textbf{\emph{CLuP$\pm$Hop model:}}}}  \hspace{.5in}  \max_{\x\in\cX(r_x,\bar{r}_x)} \pm\sqrt{\x^TG^TG\x},
\end{eqnarray}
with
\begin{eqnarray}\label{eq:clupskeq1a0}
 \cX(r_x,\bar{r}_x) \triangleq \left \{ \x | \x\in\mR^n,\|\x\|_2=r_x,\x_i^2\leq \frac{1}{n},
 \frac{1}{n}\sum_{i=1}^{n}\log\lp 1-n\x_i^2\rp =\bar{r}_x  \right \}.
\end{eqnarray}
After writing corresponding Hamiltonian and partition function
\begin{equation}
\cH_{chop}(G)= \y^TG\x,\label{eq:cska0ham1}
\end{equation}
and
\begin{equation}
Z_{chop}(\beta,G)=\sum_{\x\in\cX(r_x,\bar{r}_x)} \lp \sum_{\y\in\mS^m} e^{\beta\cH_{chop}(G)} \rp^s ,\label{eq:cska0partfun}
\end{equation}
we have for the thermodynamic limit (average) free energy
\begin{equation}
f_{chop}(\beta)=\lim_{n\rightarrow\infty}\frac{\mE_G\log{(Z_{chop}(\beta,G)})}{\beta \sqrt{n}}
=\lim_{n\rightarrow\infty} \frac{\mE_G\log\lp  \sum_{\x\in\cX(r_x,\bar{r}_x)} \lp \sum_{\y\in\mS^m} e^{\beta\cH_{chop}(G)} \rp^s   \rp  }{\beta \sqrt{n}}.\label{eq:cska0logpartfunsqrt}
\end{equation}
For the  ground state one then writes
\begin{eqnarray}
\xi(r_x,\bar{r}_x) \triangleq f_{chop}(\infty)  & \triangleq  & \lim_{\beta\rightarrow\infty}f_{chop}(\beta)  =
\lim_{\beta,n\rightarrow\infty}\frac{\mE_G\log{(Z_{chop}(\beta,G)})}{\beta \sqrt{n}}
\nonumber \\
& = &
 \lim_{n\rightarrow\infty}\frac{\mE_G \max_{\x\in\cX(r_x,\bar{r}_x)} s \max_{\y\in\mS^m}\cH_{chop}(G)}{\sqrt{n}}
 \nonumber \\
 & = &
 \lim_{n\rightarrow\infty}\frac{\mE_G \max_{\x\in\cX(r_x,\bar{r}_x)} s \max_{\y\in\mS^m}\y^TG\x }{\sqrt{n}}.
  \label{eq:cska0limlogpartfunsqrt}
\end{eqnarray}
Specializing further to $s=1$  gives
\begin{equation}
f_{chop}^+(\beta)
=\lim_{n\rightarrow\infty} \frac{\mE_G\log\lp  \sum_{\x\in\cX(r_x,\bar{r}_x)} \lp \sum_{\y\in\mS^m} e^{\beta\cH_{chop}(G)} \rp   \rp  }{\beta \sqrt{n}},\label{eq:cska0limlogpartfunsqrta0a0}
\end{equation}
and
\begin{eqnarray}
\xi^+(r_x,\bar{r}_x) \triangleq f_{chop}^+(\infty)
=
 \lim_{n\rightarrow\infty}\frac{\mE_G \max_{\x\in\cX(r_x,\bar{r}_x)}  \max_{\y\in\mS^m}\y^TG\x }{\sqrt{n}}
 =
 \lim_{n\rightarrow\infty}\frac{\mE_G \max_{\x\in\cX(r_x,\bar{r}_x)}  \sqrt{\x^TG^TG\x} }{\sqrt{n}}.
  \label{eq:cska0limlogpartfunsqrta0}
\end{eqnarray}
Analogously for $s=-1$
\begin{equation}
f_{chop}^-(\beta)
=\lim_{n\rightarrow\infty} \frac{\mE_G\log\lp  \sum_{\x\in\cX(r_x,\bar{r}_x)} \lp \sum_{\y\in\mS^m} e^{\beta\cH_{chop}(G)} \rp^{-1}  \rp  }{\beta \sqrt{n}},
  \label{eq:cska0limlogpartfunsqrta1a0}
\end{equation}
and
\begin{eqnarray}
\xi^-(r_x,\bar{r}_x) \triangleq f_{chop}^-(\infty)
& = &
 \lim_{n\rightarrow\infty}\frac{\mE_G \max_{\x\in\cX(r_x,\bar{r}_x)} - \max_{\y\in\mS^m}\y^TG\x }{\sqrt{n}}
 \nonumber \\
& = &
 \lim_{n\rightarrow\infty}\frac{\mE_G \max_{\x\in\cX(r_x,\bar{r}_x)} - \sqrt{\x^TG^TG\x} }{\sqrt{n}} =
- \lim_{n\rightarrow\infty}\frac{\mE_G \min_{\x\in\cX(r_x,\bar{r}_x)} \sqrt{\x^TG^TG\x} }{\sqrt{n}}.\nonumber \\
  \label{eq:cska0limlogpartfunsqrta1}
\end{eqnarray}
Recalling on (\ref{eq:algimpeq3})and (\ref{eq:algimpeq3a0}), one then recognizes that for a fixed $t_{0x}$ of key interest for $+$Hop is the  behavior of
\begin{equation}\label{eq:excalgimpeq7}
\bar{f}_{b}^+ \lp r_x,\bar{r}_x \rp = \min_{\x\in\bar{\cX}(r_x,\bar{r}_x)}
  - t_{0x} \|\x\|_2 - \log\lp - \lp \x^T  \lp 4 I  -  \frac{1}{n} G^TG   \rp \x - \kappa \rp \rp
- \frac{1}{n}\sum_{i=1}^{n} \log\lp 1-n\x_i^2  \rp,
\end{equation}
and for $-$Hop the  behavior of
\begin{equation}\label{eq:excalgimpeq7a0}
\bar{f}_{b}^- \lp r_x,\bar{r}_x \rp = \min_{\x\in\bar{\cX}(r_x,\bar{r}_x)}
  - t_{0x} \|\x\|_2 - \log\lp - \lp \x^T  \lp 0 I  +  \frac{1}{n} G^TG   \rp \x - \kappa \rp \rp
- \frac{1}{n}\sum_{i=1}^{n} \log\lp 1-n\x_i^2  \rp.
\end{equation}
Utilizing  (\ref{eq:clupskeq1}), (\ref{eq:cska0limlogpartfunsqrta0}), and   (\ref{eq:cska0limlogpartfunsqrta1}) we then also have
\begin{eqnarray}\label{eq:excalgimpeq8}
\bar{f}_{b}^+ \lp r_x,\bar{r}_x \rp
& = & \min_{\x\in,\bar{\cX}(r_x,\bar{r}_x )} - t_{0x}r_x - \log \lp  -4r_x^2 + \lp\bar{\xi}^+(r_x,\bar{r}_x )\rp^2 + \kappa \rp -\bar{r}_x
\nonumber \\
& = &
 - t_{0x}r_x - \log \lp  -4r_x^2 + \lp f_{chop}^+(\infty) \rp^2  + \kappa \rp - \bar{r}_x .
\end{eqnarray}
and
\begin{eqnarray}\label{eq:excalgimpeq8a0}
\bar{f}_{b}^- \lp r_x,\bar{r}_x \rp
& = & \min_{\x\in,\bar{\cX}(r_x,\bar{r}_x )} - t_{0x}r_x - \log \lp  -0r_x^2 + \lp \bar{\xi}^-(r_x,\bar{r}_x ) \rp^2 + \kappa \rp -\bar{r}_x
\nonumber \\
& = &
 - t_{0x}r_x - \log \lp  -0r_x^2 + \lp f_{chop}^-(\infty) \rp^2  + \kappa \rp - \bar{r}_x .
\end{eqnarray}
The analysis of CLuP$\pm$Hop algorithm  therefore critically depends on $f_{chop}(\infty)$ (i.e., its specialized forms $f_{chop}^+(\infty)$ and $f_{chop}^-(\infty)$). One then quickly recognizes that, in order to properly understand behavior of CLuP$\pm$Hop algorithm, properties of the associated CLuP$\pm$Hop models need to be fully understood as well. On the path to achieving that we establish a plethora of results that are of independent interest as well.

\begin{itemize}

  \item  We first connect the CLuP$\pm$Hop models to recent studies of  random processes \cite{Stojnicsflgscompyx23,Stojnicnflgscompyx23} (see introductory part of Section \ref{sec:randlincons}).

  \item Relying  on \emph{fully lifted} random duality theory (fl RDT) and its a \emph{stationarized}  sfl RDT variant  \cite{Stojnicflrdt23},  we then provide concrete   characterizations of all relevant  CLuP$\pm$Hop models features (see Sections \ref{sec:prac} and  \ref{sec:nuemrical}).

\item  Such  characterizations are connected with the CLuP$\pm$Hop algorithms in Section \ref{sec:concrete}. The algorithms' dynamics are followed through the change of three key features  -- optimal $\bar{f}_b$, $\xi$, and $r_x$. This is highlighted in Figures \ref{fig:fig5a0}-\ref{fig:fig7a0}.

    \item We observe that $\bar{f}_b$ (as a function of $r_x$) exhibits a favorable no local optima behavior which is a necessary condition for generically successful running of descending algorithms (see Figure \ref{fig:fig9a}).

    \item Due to remarkably rapid convergence of the fl RDT,  a majority of quantities of interest are usually sufficiently precisely characterized  already on the second or third level of lifting. One notable exception is characterization of configurational overlaps associated Gibbs measures cdfs. To get anywhere close to their true forms significantly higher levels of lifting need to be considered. Along the same lines, we undertake evaluations up to the 6th level and obtain cdfs that simulations closely approach (see Figures \ref{fig:fig8} and  \ref{fig:fig9} and Tables \ref{tab:2rsbunifiedsqrtpos} and \ref{tab:NEG2rsbunifiedsqrtpos}).

\item We uncover a fundamental intrinsic difference between $+$Hop and $-$Hop models. Typical near optimal configurations overlaps are well aligned for $+$Hop and almost orthogonal for $-$Hop model (see Figures \ref{fig:fig8} and  \ref{fig:fig9}).

\item The very same 6th level evaluations allow to obtain  $f_{sq}^{+,6}(\infty)\approx 1.77842$ and $f_{sq}^{-,6}(\infty)\approx -0.32807$) as $\pm$Hop model ground state free energies (see  Tables \ref{tab:2rsbunifiedsqrtpos} and \ref{tab:NEG2rsbunifiedsqrtpos}).

\item To get a better understanding as to how features of $\pm$Hop and SK models compare to each other, in Table \ref{tab:7rsbSK} and Figure \ref{fig:fig10} we  show the progress of the SK lifting mechanism up to the 7th level. We observe that overlaps of the SK model behave differently from $-$Hop model and fairly similarly to $+$Hop. Interestingly, the SK overlap cdf seems to be higher then the simple approximation given in \cite{OppShe05}. On the other hand, we obtain $f_{sq}^{(7)}(\infty)\approx 0.76319$ as the SK model ground state free energy on the 7th lifting level. This indicates that $0.76321\pm 0.00003$  prediction of  \cite{CrisRizo02} and $\approx0.76317$ prediction of    \cite{OppSch08,OppSS07} are indeed close to the true value.

\item We also practically run CLuP$\pm$Hop and show that its simulated performance  matches remarkably well obtained theoretical predictions. In particular, we observe excellent convergence and concentration properties (see Figures \ref{fig:fig5a0}-\ref{fig:figconc}).  Finally, as Table \ref{tab:tab2}
shows, for $n$ as small as few thousands CLuP$\pm$Hop  practically achieves $\pm$Hop  ground state free energies $\sim 1.774$ and $\sim 0.333$, which  very closely approaches correspo0nding  $n\rightarrow\infty$ thermodynamic limits $\approx 1.7784$ and $\approx 0.3281$. For all practical purposes this basically renders computation of $\pm$Hop models ground state free energies as \emph{typically} easy problems.

\end{itemize}

\section{Connecting CLuP$\pm$Hop model and sfl RDT}
\label{sec:randlincons}

We first note that
\begin{eqnarray}
f_{chop}(\beta) & = &\lim_{n\rightarrow\infty} \frac{\mE_G\log\lp \sum_{\x\in\cX(r_x,\bar{r}_x)} \lp  \sum_{\y\in\mS^m} e^{\beta\y^TG\x}  \rp^s  \rp } {\beta \sqrt{n}}\label{eq:hmsfl1}
\end{eqnarray}
is a function of bilinearly indexed random process  (blirp) $\y^TG\x$. Given that the machinery of  \cite{Stojnicsflgscompyx23,Stojnicnflgscompyx23} provides powerful blirps comparative mechanisms, we would like to connect  $f_{chop}(\beta)$ to it. In order to establish such a connection  several technical preliminaries are needed. For fixed positive real scalar $r_x$ and $\bar{r}_x$, consider sets $\cX(r_x,\bar{r}_x)\subseteq \mR^n$ and  $\cY\triangleq\mS^m$ (with $\mS^m$ be the unit sphere in $\mR^m$). Let $f_S(\cdot):\mR^n\rightarrow R$ be a given function and for $r\in\mN$, $k\in\{1,2,\dots,r+1\}$ let vectors $\p=[\p_0,\p_1,\dots,\p_{r+1}]$, $\q=[\q_0,\q_1,\dots,\q_{r+1}]$, and $\c=[\c_0,\c_1,\dots,\c_{r+1}]$ be such that
 \begin{eqnarray}\label{eq:hmsfl2}
1=\p_0\geq \p_1\geq \p_2\geq \dots \geq \p_r\geq \p_{r+1} & = & 0
\nonumber \\
1=\q_0\geq \q_1\geq \q_2\geq \dots \geq \q_r\geq \q_{r+1} & = & 0,
 \end{eqnarray}
$\c_0=1$, $\c_{r+1}=0$. Also let the elements of  $u^{(4,k)}\in\mR$, $\u^{(2,k)}\in\mR^m$, and $\h^{(k)}\in\mR^n$ be independent standard normals and set ${\mathcal U}_k\triangleq [u^{(4,k)},\u^{(2,k)},\h^{(k)}]$ . Moreover, set
  \begin{eqnarray}\label{eq:fl4}
\psi_{S,\infty}(f_{S},\calX,\calY,\p,\q,\c,r_x,\bar{r}_x)  \triangleq
 \mE_{G,{\mathcal U}_{r+1}} \frac{1}{n\c_r} \log
\lp \mE_{{\mathcal U}_{r}} \lp \dots \lp \mE_{{\mathcal U}_2}\lp\lp\mE_{{\mathcal U}_1} \lp \lp Z_{S,\infty}\rp^{\c_2}\rp\rp^{\frac{\c_3}{\c_2}}\rp\rp^{\frac{\c_4}{\c_3}} \dots \rp^{\frac{\c_{r}}{\c_{r-1}}}\rp,
 \end{eqnarray}
with
\begin{eqnarray}\label{eq:fl5}
Z_{S,\infty} & \triangleq & e^{D_{0,S,\infty}} \nonumber \\
 D_{0,S,\infty} & \triangleq  & \max_{\x\in\cX,\|\x\|_2=r_x}  s  \max_{\x\in\cY,\|\y\|_2=\bar{r}_x}
 \lp f_{S}
+\sqrt{n}  r_y    \lp\sum_{k=2}^{r+1}c_k\h^{(k)}\rp^T\x
+ \sqrt{n} r_x \y^T\lp\sum_{k=2}^{r+1}b_k\u^{(2,k)}\rp \rp
\nonumber  \\
 b_k & \triangleq & b_k(\p)=\sqrt{\p_{k-1}-\p_k}
 \nonumber  \\
 c_k & \triangleq & c_k(\q)=\sqrt{\q_{k-1}-\q_k}.
 \end{eqnarray}
The following theorem is a fundamental sfl RDT result.
\begin{theorem} [\cite{Stojnicflrdt23}]
\label{thm:thmsflrdt1}  Assume linear/proportional large  $n$  regime where  $\alpha=\lim_{n\rightarrow\infty} \frac{m}{n}$, remains constant as  $n$ grows. Let the elements of  $G\in\mR^{m\times n}$ be independent standard normals  and assume the complete sfl RDT frame from \cite{Stojnicsflgscompyx23}.  For three given positive real scalars $r_x$, $\bar{r}_x$, and $r_y$
let $\cX(r_x,\bar{r}_x)\in\mR^n$ and $\cY\in\mR^m$ be two sets with norm of their elements being $r_x$ and $r_y$, respectively.  For a given function $f(\x):R^n\rightarrow R$ set
\begin{align}\label{eq:thmsflrdt2eq1}
   \psi_{rp}(r_x,\bar{r}_x) & \triangleq - \max_{\x\in\cX(r_x,\bar{r}_x)} s \max_{\y\in\cY} \lp f(\x)+\y^TG\x \rp
   \qquad  \mbox{(\bl{\textbf{random primal}})} \nonumber \\
   \psi_{rd}(\p,\q,\c,r_x,\bar{r}_x) & \triangleq    \frac{r_x^2r_y^2}{2}    \sum_{k=2}^{r+1}\Bigg(\Bigg.
   \p_{k-1}\q_{k-1}
   -\p_{k}\q_{k}
  \Bigg.\Bigg)
\c_k
\nonumber\\
&  \hspace{.15in}  - \psi_{S,\infty}(f(\x),\calX(r_x,\bar{r}_x),\calY,\p\q,\c,r_x,\bar{r}_x) \hspace{.24in} \mbox{(\bl{\textbf{fl random dual}})}.
 \end{align}
Let $\hat{\p_0}\rightarrow 1$,$\hat{\q_0}\rightarrow 1$,$\hat{\c_0}\rightarrow 1$, $\hat{\p}_{r+1}=\hat{\q}_{r+1}=\hat{\c}_{r+1}=0$, and let the non-fixed parts of $\hat{\p}\triangleq \hat{\p}(r_x,\bar{r}_x,r_y)$, $\hat{\q}\triangleq \hat{\q}(r_x,\bar{r}_x,r_y)$, and  $\hat{\c}\triangleq \hat{\c}(r_x,\bar{r}_x,r_y)$ be the solutions of the following system
\begin{eqnarray}\label{eq:thmsflrdt2eq2}
    \frac{d \psi_{rd}(\p,\q,\c,r_x,\bar{r}_x)}{d\p} =  0,\quad
    \frac{d \psi_{rd}(\p,\q,\c,r_x,\bar{r}_x)}{d\q} =  0,\quad
   \frac{d \psi_{rd}(\p,\q,\c,r_x,\bar{r}_x)}{d\c} =  0.
 \end{eqnarray}
 Then,
\begin{eqnarray}\label{eq:thmsflrdt2eq3}
    \lim_{n\rightarrow\infty} \frac{\mE_G  \psi_{rp}}{\sqrt{n}}
  & = &
 \lim_{n\rightarrow\infty} \psi_{rd}(\hat{\p}(r_x,\bar{r}_x,r_y),\hat{\q}(r_x,\bar{r}_x,r_y),\hat{\c}(r_x,\bar{r}_x,r_y),r_x) \qquad \mbox{(\bl{\textbf{strong sfl random duality}})},\nonumber \\
 \end{eqnarray}
where $\psi_{S,\infty}(\cdot)$ is as in (\ref{eq:fl4})-(\ref{eq:fl5}).
 \end{theorem}
\begin{proof}
Follows immediately from \cite{Stojnicflrdt23,Stojnicsflgscompyx23} after trivial cosmetic changes in the definition of set $\calX$.
 \end{proof}

\subsection{Practical relevance}
\label{sec:prac}

To make use of the above theorem, we take $f(\x)=0$, $\calY=\mS^m$ (with $\mS^m$ being the unit sphere in $\mR^m$), and assume that $\calX(r_x,\bar{r}_x)$ is as given in (\ref{eq:clupskeq1a0}). Noting that $r_y=1$ we then recognize that  the key object of practical interest is the following so-called \emph{random dual}
\begin{align}\label{eq:prac1}
    \psi_{rd}(\p,\q,\c,r_x,\bar{r}_x) & \triangleq    \frac{r_x^2}{2}    \sum_{k=2}^{r+1}\Bigg(\Bigg.
   \p_{k-1}\q_{k-1}
   -\p_{k}\q_{k}
  \Bigg.\Bigg)
\c_k
  - \psi_{S,\infty}(0,\calX(r_x,\bar{r}_x),\calY,\p,\q,\c,r_x). \nonumber \\
  & =   \frac{1}{2}    \sum_{k=2}^{r+1}\Bigg(\Bigg.
   \p_{k-1} \q_{k-1}
   -\p_{k}\q_{k}
  \Bigg.\Bigg)
\c_k
  - \frac{1}{n}\varphi(D^{(bin)}(r_x,\bar{r}_x))- \frac{1}{n}\varphi(D^{(sph)}(r_x,\bar{r}_x)),
  \end{align}
where analogously to (\ref{eq:fl4})-(\ref{eq:fl5})
  \begin{eqnarray}\label{eq:prac2}
\varphi(D,\c) & = &
 \mE_{G,{\mathcal U}_{r+1}} \frac{1}{\c_r} \log
\lp \mE_{{\mathcal U}_{r}} \lp \dots \lp \mE_{{\mathcal U}_3}\lp\lp\mE_{{\mathcal U}_2} \lp
\lp    e^{D}   \rp^{\c_2}\rp\rp^{\frac{\c_3}{\c_2}}\rp\rp^{\frac{\c_4}{\c_3}} \dots \rp^{\frac{\c_{r}}{\c_{r-1}}}\rp,
  \end{eqnarray}
and
\begin{eqnarray}\label{eq:prac3}
D^{(bin)}(r_x,\bar{r}_x) & = & \max_{\x\in\cX(r_x,\bar{r}_x)} \lp   s\sqrt{n} r_y      \lp\sum_{k=2}^{r+1}c_k\h^{(k)}\rp^T\x  \rp
=
\max_{\x\in\cX (r_x,\bar{r}_x) } \lp   s\sqrt{n}      \lp\sum_{k=2}^{r+1}c_k\h^{(k)}\rp^T\x  \rp
\nonumber \\
  D^{(sph)}(r_x,\bar{r}_x) & =  &   s \max_{\y\in\cY}
\lp \sqrt{n} r_x \y^T\lp\sum_{k=2}^{r+1}b_k\u^{(2,k)}\rp \rp
=
s r_x  \max_{\y\in\mS^m}
\lp \sqrt{n} \y^T\lp\sum_{k=2}^{r+1}b_k\u^{(2,k)}\rp \rp.
 \end{eqnarray}
From \cite{Stojnicclupsk25}'s (75) and (76) one finds
\begin{eqnarray}\label{eq:excprac4}
D^{(bin)}(r_x,\bar{r}_x)  =  \max_{\x\in\bar{\cX}(r_x,\bar{r}_x)} \lp s \sqrt{n}      \lp\sum_{k=2}^{r+1}c_k\h^{(k)}\rp^T\x  \rp
 = \min_{\gamma,\nu} \lp -  \sum_{i=1}^n D^{(bin)}_i(c_k) +\gamma r_x^2 n +\nu \bar{r}_x n \rp,
 \end{eqnarray}
where scaling $\gamma\sim \gamma\sqrt{n}$ and $\nu\sim \nu\sqrt{n}$ is adopted and
\begin{eqnarray}\label{eq:excprac5}
D^{(bin)}_i(c_k)=    - s\lp\sum_{k=2}^{r+1}c_k\h_i^{(k)}\rp \x_i +\gamma\x_i^2 + \nu\log(1-n\x_i^2).
\end{eqnarray}
On the other hand, from \cite{Stojnichopflrdt23}'s (29)-(30), we have
\begin{eqnarray}\label{eq:prac9}
   D^{(sph)}(r_x,\bar{r}_x)
  & =  &   s r_x \min_{\gamma_{sq}} \lp \sum_{i=1}^{m} D_i^{(sph)}(b_k)+\gamma_{sq}n \rp,
 \end{eqnarray}
with
\begin{eqnarray}\label{eq:prac10}
   D_i^{(sph)}(b_k)= \frac{\lp \sum_{k=2}^{r+1}b_k\u_i^{(2,k)}  \rp^2}{4\gamma_{sq}}.
 \end{eqnarray}
After connecting $\xi(r_x,\bar{r}_x)$ and $f_{chop}(\infty)$ from (\ref{eq:cska0limlogpartfunsqrt}) and random primal, $\psi_{rp}(r_x,\bar{r}_x)$, from Theorem \ref{thm:thmsflrdt1}, we write
 \begin{eqnarray}
\xi(r_x,\bar{r}_x) = f_{chop}(\infty)
& = &  \lim_{n\rightarrow\infty}\frac{\mE_G \max_{\x\in\cX(r_x,\bar{r}_x)} s\max_{\y\in\mS^m} \y^TG\x  }{\sqrt{n}}
\nonumber \\
 & = &
      -\lim_{n\rightarrow\infty} \frac{\mE_G  \psi_{rp}}{\sqrt{n}}
   =
 -  \lim_{n\rightarrow\infty} \psi_{rd}(\hat{\p},\hat{\q},\hat{\c},r_x,\bar{r}_x),
  \label{eq:prac11}
\end{eqnarray}
where the non-fixed parts of $\hat{\p}$, $\hat{\q}$, and  $\hat{\c}$ satisfy
\begin{eqnarray}\label{eq:prac12}
    \frac{d \psi_{rd}(\p,\q,\c,r_x,\bar{r}_x)}{d\p} =  0,\quad
    \frac{d \psi_{rd}(\p,q,\c,r_x,\bar{r}_x)}{d\q} =  0,\quad
   \frac{d \psi_{rd}(\p,\q,\c,r_x,\bar{r}_x)}{d\c} =  0.
 \end{eqnarray}
 Relying on (\ref{eq:prac1})-(\ref{eq:prac12}), we further have
 \begin{eqnarray}
 \lim_{n\rightarrow\infty} \psi_{rd}(\hat{\p},\hat{\q},\hat{\c},r_x,\bar{r}_x) =  \bar{\psi}_{rd}(\hat{\p},\hat{\q},\hat{\c},\hat{\gamma},\hat{\nu},\hat{\gamma}_{sq},r_x,\bar{r}_x),
  \label{eq:prac12a}
\end{eqnarray}
where
\begin{eqnarray}\label{eq:prac13}
    \bar{\psi}_{rd}(\p,\q,\c,\gamma, \nu,\gamma_{sq},r_x,\bar{r}_x)   & = &  \frac{1}{2}    \sum_{k=2}^{r+1}\Bigg(\Bigg.
   \p_{k-1}\q_{k-1}
   - \p_{k}\q_{k}
  \Bigg.\Bigg)
\c_k
-\gamma r_x^2-\nu\bar{r}_x
- \varphi(D_1^{(bin)}(c_k(\q)),\c)
\nonumber \\
& &
-s\gamma_{sq} r_x - \alpha \varphi(D_1^{(sph)}(b_k(\p)),\c).
  \end{eqnarray}
Based on (\ref{eq:prac2}), (\ref{eq:excprac5}), and  (\ref{eq:prac10}), $\varphi(D_1^{(bin)}(c_k(\q)),\c)$  and $\varphi(D_1^{(sph)}(b_k(\p)),\c)$  in (\ref{eq:prac13})  are  given by
\begin{align}\label{eq:prac14}
\varphi(D_1^{(bin)}(c_k(\q)),\c) & =
 \mE_{{\mathcal U}_{r+1}} \frac{1}{\c_r} \log
\lp \mE_{{\mathcal U}_{r}} \lp \dots \lp \mE_{{\mathcal U}_3}\lp\lp\mE_{{\mathcal U}_2} \lp
    e^{  -\c_2 D_1^{(bin)}(c_k(\q))  }  \rp\rp^{\frac{\c_3}{\c_2}}\rp\rp^{\frac{\c_4}{\c_3}} \dots \rp^{\frac{\c_{r}}{\c_{r-1}}}\rp,
    \nonumber \\
\varphi(D_1^{(sph)}(b_k(\p)),\c) & =
 \mE_{{\mathcal U}_{r+1}} \frac{1}{\c_r} \log
\lp \mE_{{\mathcal U}_{r}} \lp \dots \lp \mE_{{\mathcal U}_3}\lp\lp\mE_{{\mathcal U}_2} \lp
    e^{  s r_x \c_2 D_1^{(sph}(b_k(\p))  }  \rp\rp^{\frac{\c_3}{\c_2}}\rp\rp^{\frac{\c_4}{\c_3}} \dots \rp^{\frac{\c_{r}}{\c_{r-1}}}\rp,
  \end{align}
whereas $\hat{\gamma}$, $\hat{\nu}$, $\hat{\gamma}_{sq}$, and  the non-fixed parts of $\hat{\p}$, $\hat{\q}$, and  $\hat{\c}$ are the solutions of
\begin{eqnarray}\label{eq:prac16}
    \frac{d \bar{\psi}_{rd}(\p,\q,\c,\gamma,\nu,\gamma_{sq},r_x,\bar{r}_x)}{d\p}  =
    \frac{d \bar{\psi}_{rd}(\p,\q,\c,\gamma,\nu,\gamma_{sq},r_x,\bar{r}_x)}{d\q}  =
   \frac{d \bar{\psi}_{rd}(\p,\q,\c,\gamma,\nu,\gamma_{sq},r_x,\bar{r}_x)}{d\c} & = & 0 \nonumber \\
      \frac{d \bar{\psi}_{rd}(\p,\q,\c,\gamma,\nu,\gamma_{sq},r_x,\bar{r}_x)}{d\gamma}  =
            \frac{d \bar{\psi}_{rd}(\p,\q,\c,\gamma,\nu,\gamma_{sq},r_x,\bar{r}_x)}{d\nu}  =
   \frac{d \bar{\psi}_{rd}(\p,\q,\c,\gamma,\nu,\gamma_{sq},r_x,\bar{r}_x)}{d\gamma_{sq}} & = & 0.\nonumber \\
 \end{eqnarray}
After  observing
\begin{eqnarray}\label{eq:prac17}
b_k(\hat{\p})  & = & \sqrt{\hat{\p}_{k-1}-\hat{\p}_k} \nonumber \\
c_k(\hat{\q})  & = & \sqrt{\hat{\q}_{k-1}-\hat{\q}_k},
 \end{eqnarray}
and connecting  (\ref{eq:prac11}) to (\ref{eq:prac12a}) one further finds
 \begin{eqnarray}
f_{chop}(\infty)
& = &  \lim_{n\rightarrow\infty}\frac{\mE_G \max_{\x\in\cX(r_x,\bar{r}_x)} s \max_{\t\in\mS^m}   \y^TG\x}{\sqrt{n}}
\nonumber \\
& = &
 - \lim_{n\rightarrow\infty} \psi_{rd}(\hat{\p},\hat{\q},\hat{\c},r_x,\bar{r}_x)
 = -  \bar{\psi}_{rd}(\hat{\p},\hat{\q},\hat{\c},\hat{\gamma},\hat{\nu},\hat{\gamma}_{sq},r_x,\bar{r}_x) \nonumber \\
 & = &     -\frac{1}{2}    \sum_{k=2}^{r+1}\Bigg(\Bigg.
    \hat{\p}_{k-1}\hat{\q}_{k-1}
   - \hat{\p}_{k}\hat{\q}_{k}
  \Bigg.\Bigg)
\hat{\c}_k
 + \hat{\gamma} r_x^2 + \hat{\nu}\bar{r}_x
  + \varphi(D_1^{(bin)}(c_k(\hat{\q})),\c)
\nonumber \\
& &
   + s\hat{\gamma}_{sq}  r_x
  + \alpha \varphi(D_1^{(sph)}(b_k(\hat{\p})),\c).
  \label{eq:prac18}
\end{eqnarray}
We summarize the above results in the following theorem.

\begin{theorem}
  \label{thme:thmprac1}
Consider linear/proportional  large $n$  regime with $\alpha=\lim_{n\rightarrow\infty} \frac{m}{n}$ and  assume the complete sfl RDT setup of \cite{Stojnicsflgscompyx23}. Let the ``fixed'' parts of $\hat{\p}$, $\hat{\q}$, and $\hat{\c}$ be $\hat{\p}_1\rightarrow 1$, $\hat{\q}_1\rightarrow 1$, $\hat{\c}_1\rightarrow 1$, $\hat{\p}_{r+1}=\hat{\q}_{r+1}=\hat{\c}_{r+1}=0$. For $\varphi(\cdot)$ and $\bar{\psi}_{rd}(\cdot)$ from (\ref{eq:prac2}) and (\ref{eq:prac13}) let $\hat{\gamma}$, $\hat{\nu}$, $\hat{\gamma}_{sq}$, and the ``non-fixed'' parts of $\hat{\p}_k$, $\hat{\q}_k$, and $\hat{\c}_k$ ($k\in\{2,3,\dots,r\}$) be the solutions of (\ref{eq:prac16}). Taking $b_k(\hat{\p})$ and $c_k(\hat{\q})$  as in (\ref{eq:prac17}) gives
 \begin{eqnarray}
\xi(r_x,\bar{r}_x) \triangleq f_{chop}(\infty)
& = &   -  \bar{\psi}_{rd}(\hat{\p},\hat{\q},\hat{\c},\hat{\gamma},\hat{\nu},\hat{\gamma}_{sq},r_x,\bar{r}_x) \nonumber \\
 & = &     -\frac{1}{2}    \sum_{k=2}^{r+1}\Bigg(\Bigg.
    \hat{\p}_{k-1}\hat{\q}_{k-1}
   - \hat{\p}_{k}\hat{\q}_{k}
  \Bigg.\Bigg)
\hat{\c}_k
 + \hat{\gamma} r_x^2 + \hat{\nu}\bar{r}_x
  + \varphi(D_1^{(bin)}(c_k(\hat{\q})),\c)
\nonumber \\
& &
   + s\hat{\gamma}_{sq}  r_x
  + \alpha \varphi(D_1^{(sph)}(b_k(\hat{\p})),\c).
  \label{eq:thmprac1eq1}
\end{eqnarray}
\end{theorem}
\begin{proof}
Follows directly from (\ref{eq:prac18}) as an immediate consequence of the above discussion and Theorem \ref{thm:thmsflrdt1}.
\end{proof}

\subsection{Numerical evaluations}
\label{sec:nuemrical}

Theorem \ref{thme:thmprac1} is practically relevant as long as  all the underlying numerical evaluations cam be conducted. We show next how this is done. It turns out as convenient to work with the third partial level of lifting (we denote it as 3-spl RDT; in general, for the $r$-th level we follow the convention of \cite{Stojnichopflrdt23} and denote partial and full lifting by $r$-spl and $r$-sfl, respectively). Such a choice makes a nice balance between accuracy and numerical burden (for example, recalling on \cite{Stojnichopflrdt23}, one has that on the third partial level the limiting $\pm$Hop ground state energies are approached with accuracy $\sim 0.001$).

To put everything on concrete terms, we first note that on the third partial level   $r=3$ and $\hat{\p}_1\rightarrow 1$, $\hat{\q}_1\rightarrow 1$, and  $\hat{\p}_{r+1}=\hat{\p}_{4}=\hat{\q}_{r+1}=\hat{\q}_{4}=0$. In addition to having  $\hat{\c}_{2}\neq 0$, $\hat{\c}_{3}\neq 0$, $\p_2\neq0$, and $\q_2\neq0$ one also has $\p_3=\q_3=0$. Keeping all of this in mind we can then write
\begin{align}\label{eq:prac26}
   - \bar{\psi}_{rd}^{(3,p)} (\p,\q,\c,\gamma,\nu,\gamma_{sq},r_x,\bar{r}_x)   & =  - \frac{1}{2}
(1-\p_2\q_2)\c_2r_x^2 - \frac{1}{2}
\p_2\q_2\c_3r_x^2
\nonumber \\
&
   + \gamma r_x^2 +\nu\bar{r}_x
  + \frac{1}{\c_3} \log\lp   \mE_{{\mathcal U}_3}   \lp   \mE_{{\mathcal U}_2} e^{   -\c_2 D_1^{(bin)}(c_k(\q))    }  \rp^{\frac{\c_3}{\c_2}}   \rp \nonumber \\
&
   + s \gamma r_x
     + \frac{\alpha}{\c_3} \log\lp   \mE_{{\mathcal U}_3}   \lp   \mE_{{\mathcal U}_2} e^{   s r_x \c_2 D_1^{(sph)}(b_k(\p))    }  \rp^{\frac{\c_3}{\c_2}}   \rp   \nonumber \\
& =  - \frac{1}{2}
(1-\p_2\q_2)\c_2r_x^2
 - \frac{1}{2}
\p_2\q_2\c_3r_x^2
   + \gamma r_x^2 + \nu\bar{r}_x
  + f_q^{(3,p)} + s \gamma_{sq} r_x + \alpha f_{q,s}^{(3,p)},
   \end{align}
where
\begin{eqnarray} \label{eq:prac26a0}
f_q^{(3,p)} = \frac{1}{\c_3} \log\lp   \mE_{{\mathcal U}_3}   \lp   \mE_{{\mathcal U}_2} e^{   - \c_2 D_1^{(bin)}(c_k(\q))    }  \rp^{\frac{\c_3}{\c_2}}   \rp,
 \end{eqnarray}
 and based on \cite{Stojnichopflrdt23}'s (48) and (90)
\begin{eqnarray} \label{eq:prac26a0a0}
f_{q,s}^{(3,p)}& = & \frac{\alpha}{\c_3} \log\lp   \mE_{{\mathcal U}_3}   \lp   \mE_{{\mathcal U}_2} e^{   s r_x \c_2 D_1^{(sph)}(b_k(\p))    }  \rp^{\frac{\c_3}{\c_2}}   \rp
\nonumber \\
& = &
  \alpha
\Bigg(\Bigg. -\frac{1}{2\c_2} \log \lp \frac{2s\gamma_{sq}-\c_2r_x(1-\p_2)}{2s\gamma_{sq}} \rp   -\frac{1}{2\c_3} \log \lp \frac{2s\gamma_{sq}-\c_2 r_x (1-\p_2)-\c_3r_x\p_2 }  {2s\gamma_{sq}-\c_2r_x(1-\p_2)} \rp     \Bigg.\Bigg).
 \end{eqnarray}
The above provides all necessary ingredients to conduct the concrete numerical evaluations. We show next what results such evaluations produce.

\subsubsection{Concrete results}
\label{sec:concrete}

Recalling on (\ref{eq:excalgimpeq8}) and  (\ref{eq:excalgimpeq8a0}) we first write
\begin{eqnarray}\label{eq:concexcalgimpeq8}
 \bar{f}_{b}^+ \lp r_x,\bar{r}_x \rp
& = & \min_{\x\in,\bar{\cX}(r_x,\bar{r}_x )} - t_{0x}r_x - \log \lp  -4r_x^2 + \lp \xi^+(r_x,\bar{r}_x )\rp^2 + \kappa \rp -\bar{r}_x
\nonumber \\
& = &
 - t_{0x}r_x - \log \lp  -4r_x^2 + \lp f_{chop}^+(\infty) \rp^2 + \kappa \rp - \bar{r}_x,
\end{eqnarray}
and
\begin{eqnarray}\label{eq:concexcalgimpeq8a0}
 \bar{f}_{b}^- \lp r_x,\bar{r}_x \rp
& = & \min_{\x\in,\bar{\cX}(r_x,\bar{r}_x )} - t_{0x}r_x - \log \lp  -0r_x^2 + \lp \xi^-(r_x,\bar{r}_x )\rp^2  + \kappa \rp -\bar{r}_x
\nonumber \\
& = &
 - t_{0x}r_x - \log \lp  -0r_x^2 + \lp f_{chop}^-(\infty) \rp^2 + \kappa \rp - \bar{r}_x .
\end{eqnarray}
 After setting
\begin{eqnarray}\label{eq:concexcalgimpeq9a0a0}
\bar{r}_x^{(opt,\pm)}  \triangleq  \mbox{argmin}_{\bar{r}_x <0} \bar{f}_{b}^{\pm} \lp r_x,\bar{r}_x \rp,
\end{eqnarray}
and
\begin{eqnarray}\label{eq:excalgimpeq9a0}
\bar{f}_{b}^{\pm} \lp r_x \rp
 &  \triangleq &  \min_{\bar{r}_x <0} \bar{f}_{b}^{\pm} \lp r_x,\bar{r}_x \rp  =  \bar{f}_{b}^{\pm} \lp r_x,\bar{r}_x^{(opt,\pm)} \rp \nonumber \\
\xi^{\pm} \lp r_x \rp
 &  \triangleq &    \xi^{\pm} \lp r_x,\bar{r}_x^{(opt,\pm)} \rp,
\end{eqnarray}
we have for the optimal $r_x$
\begin{eqnarray}\label{eq:excalgimpeq9}
\hat{r}_x^{\pm} =  \mbox{argmin}_{r_x\in(0,1]} \bar{f}_{b}^{\pm} \lp r_x \rp.
\end{eqnarray}

After conducting all numerical evaluations  we obtained results presented in Figures \ref{fig:fig5a0}-\ref{fig:fig7a0}.  Theoretical predictions for the dynamics of all three critical quantities, $\frac{\bar{f}_b(\hat{r}_x)}{t_{0x}}$, $\xi(\hat{r}_x)$, and $\hat{r}_x$  are shown ($\pm$ signs are skipped in all figures to lighten notation; clearly for $+$Hop model, all quantities are with $+$ sign and for $-$Hop with $-$ sign). Also, for the very same quantities, we in parallel show the simulated results obtained by running CLuP$\pm$Hop algorithm for $n=200$, $n=1000$, and $n=4000$. Despite the simulated dimensions being rather small (compared to $n\rightarrow \infty$), an excellent agreement between theoretical and simulated results is observed.

We should also add another interesting point. Namely, all of the above numerical evaluations were repeated while relying on modulo-$\m$ sfl results from \cite{Stojnicsflgscompyx23,Stojnicnflgscompyx23,Stojnicflrdt23} with no difference in the obtained results. This indicates the  \emph{minimization} type of $\c$ \emph{stationarity} and is in a remarkable agreement with the same type of discovery from \cite{Stojnichopflrdt23,Stojnicbinperflrdt23,Stojnicnegsphflrdt23,Stojnicabple25,Stojnicclupsk25}.

\begin{figure}[h]
\centering
\centerline{\includegraphics[width=1.00\linewidth]{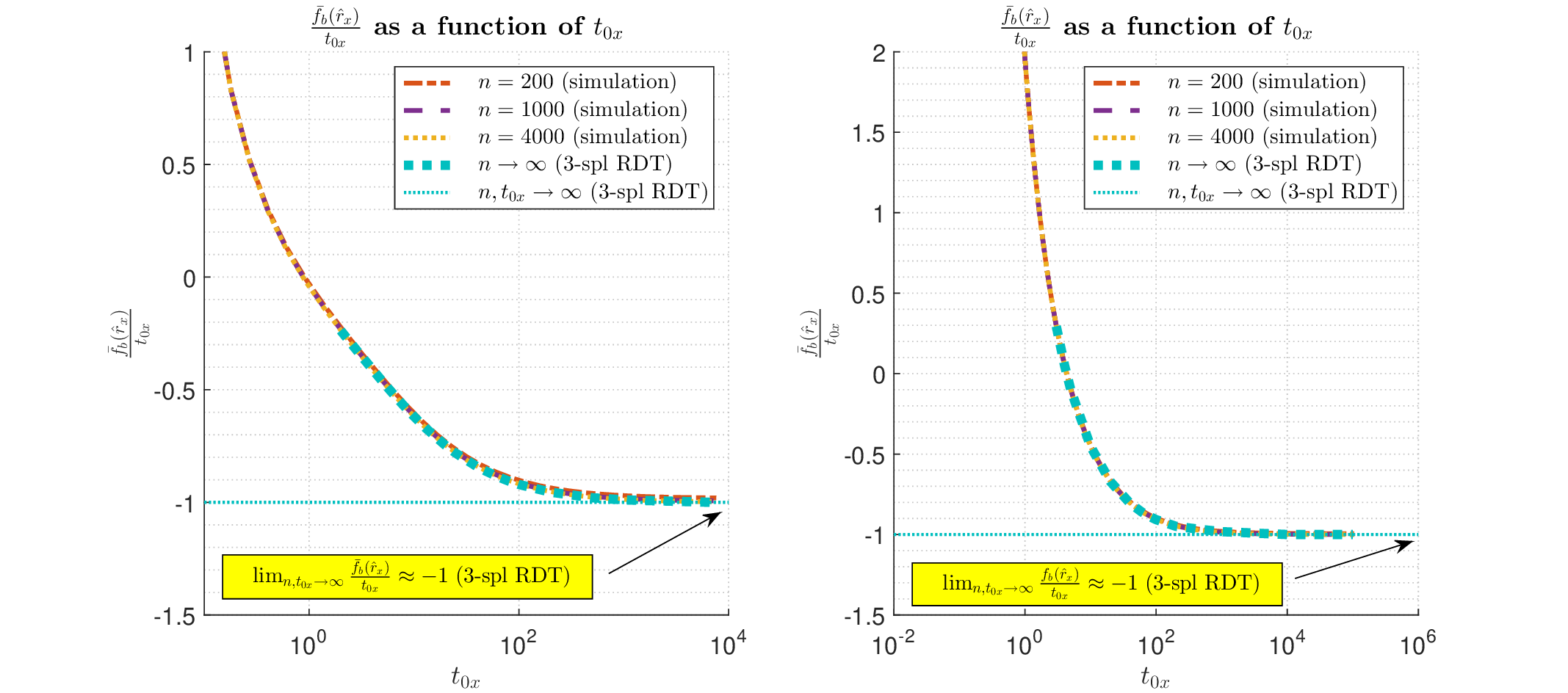}}
\caption{$\frac{\bar{f}_b(\hat{r}_x)}{t_{0x}}$ as a function of $t_{0x}$;  $+$Hop (left) and $-$Hop (right) }
\label{fig:fig5a0}
\end{figure}
\begin{figure}[h]
\centering
\centerline{\includegraphics[width=1.00\linewidth]{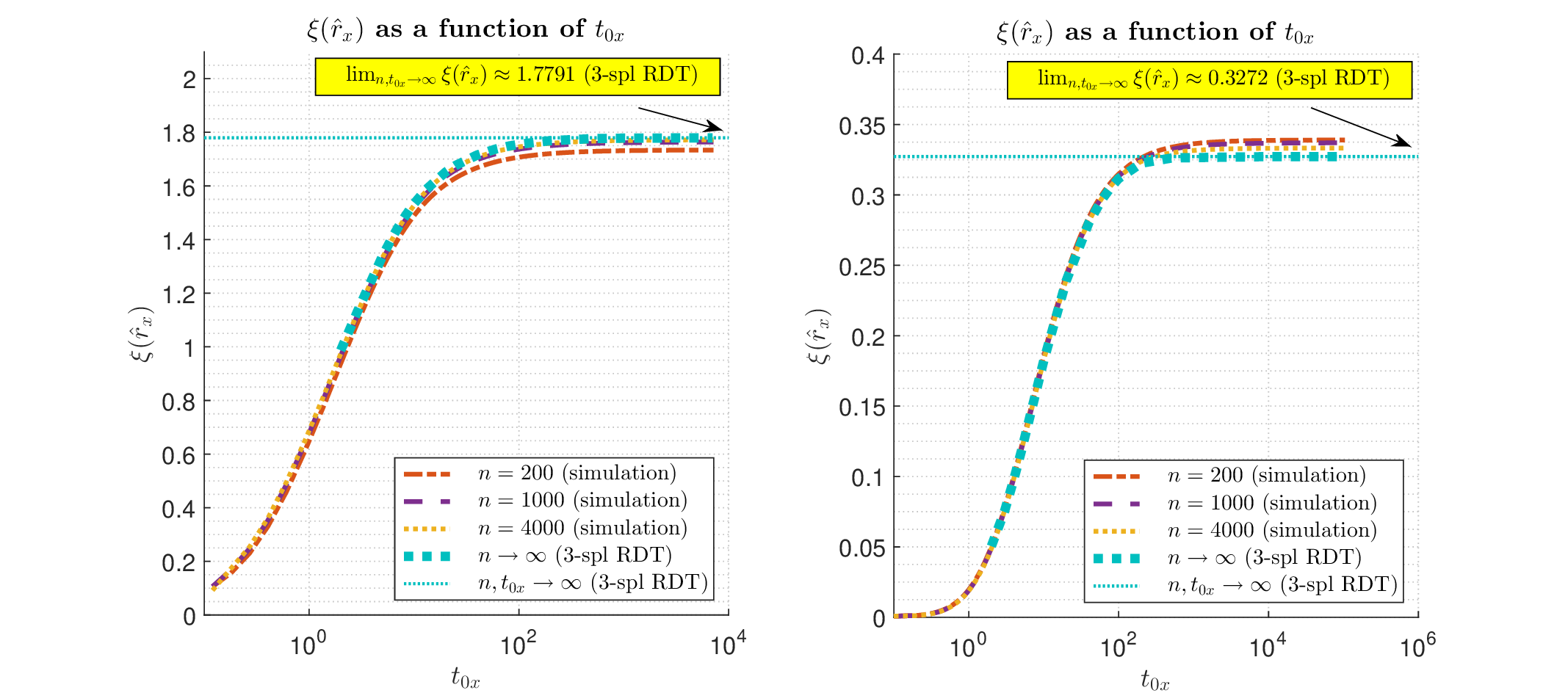}}
\caption{$\bar{\xi}(\hat{r}_x)$ as a function of $t_{0x}$;  $+$Hop (left) and $-$Hop (right) }
\label{fig:fig6a0}
\end{figure}
\begin{figure}[h]
\centering
\centerline{\includegraphics[width=1.00\linewidth]{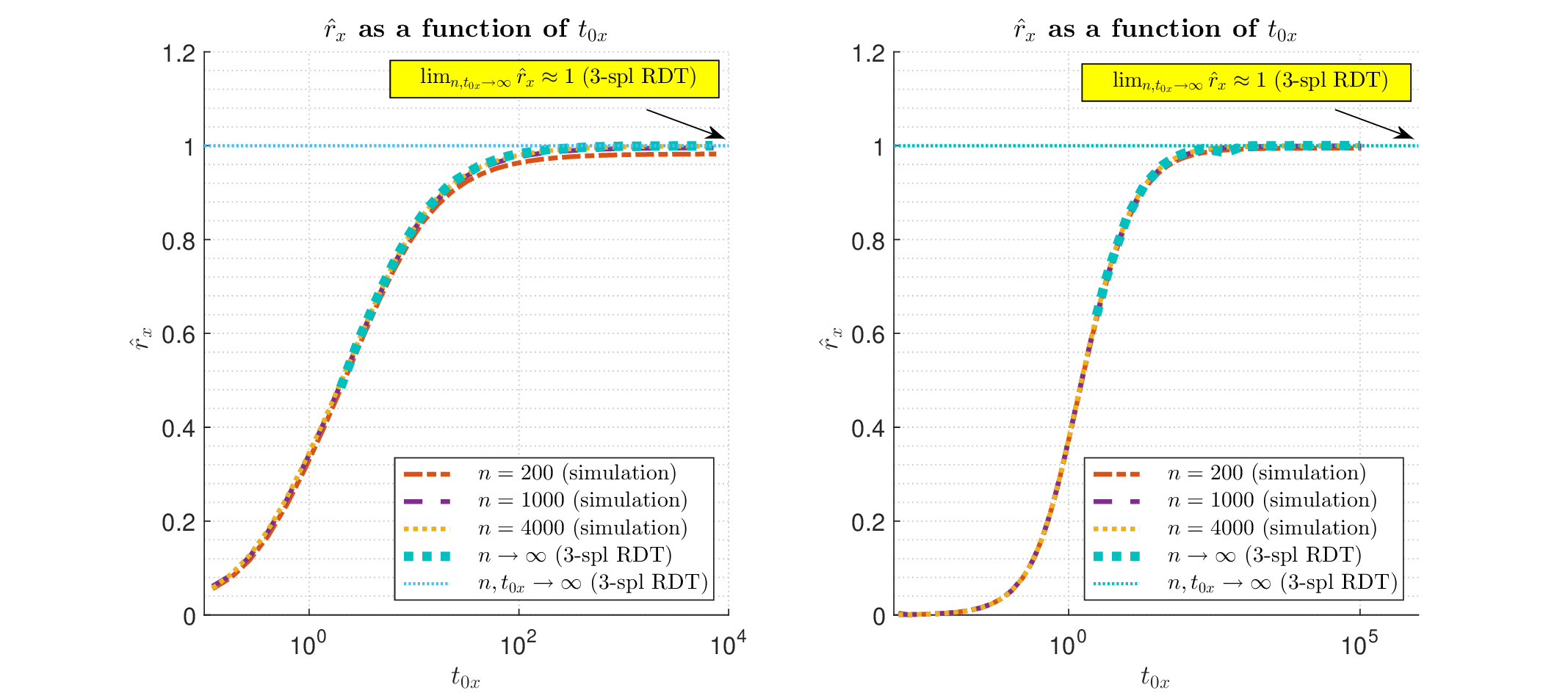}}
\caption{$\hat{r}_x$ as a function of $t_{0x}$;  $+$Hop (left) and $-$Hop (right)}
\label{fig:fig7a0}
\end{figure}

\subsubsection{Properties}
\label{sec:properties}

The conducted experiments allowed also to look at several interesting properties in more detail.

\noindent \underline{\textbf{\emph{Convergence with respect to $n$}}}

The theoretical results are obtained in the thermodynamic limit with $n\rightarrow\infty$. Since only finite $n$'s can be simulated, it is interesting to see how quickly simulated results converge as $n$ increases. We focus on key quantity of interest, $\xi(\hat{r}_x)$, and show in Figure \ref{fig:figconv1} how it changes as a function of $n$. Starting with $n$ as small as a few tens and then increasing it to reach  a few thousands  we observe the pace at which the simulated CLuP$\pm$Hop dynamics approach theoretical predictions. For $+$Hop the most rapid portion of the convergence process happens for $n\leq 1000$ and then it visibly slows down. On the other hand, for $-$Hop after starting slowly for $n\sim 100$ and the convergence speeds up for $n\sim 1000$ before naturally  slowing down as $n$ gets to the range of several thousands. Such $-$Hop behavior is somewhat interesting and it should be noted that it is in an excellent agreement with observations made in \cite{StojnicHopBnds10}.

In Figure  \ref{fig:figconv2} we show how  $\lim_{t_{0x}\rightarrow\infty}\xi(\hat{r}_x)$  changes with $n$ ($\lim_{t_{0x}\rightarrow\infty}\xi(\hat{r}_x)$ is the limiting value of  $\xi(\hat{r}_x)$ obtained at the end of the algorithm's running).  $\lim_{t_{0x}\rightarrow\infty}\hat{\xi}(\hat{r}_x)$ -- the effective objective value of (\ref{eq:inteq1}) produced by the algorithm and an emulation of $\lim_{t_{0x}\rightarrow\infty}\xi(\hat{r}_x)$ obtained after $\pm \frac{1}{\sqrt{n}}$ rounding of the algorithm's output $\hat{\x}$ -- is shown in parallel as well. Clearly,  for $n\rightarrow\infty$ the two quantities $\lim_{t_{0x}\rightarrow\infty}\xi(\hat{r}_x)$ and $\lim_{t_{0x}\rightarrow\infty}\hat{\xi}(\hat{r}_x)$ are indistinguishable. However, for finite $n$ they do differ and one expects the difference to fade away as $n$ increases. As  Figure  \ref{fig:figconv2} shows, this indeed happens. Moreover, it happens already for $n$ as small sa a few thousands.

\begin{figure}[h]
\centering
\centerline{\includegraphics[width=1.00\linewidth]{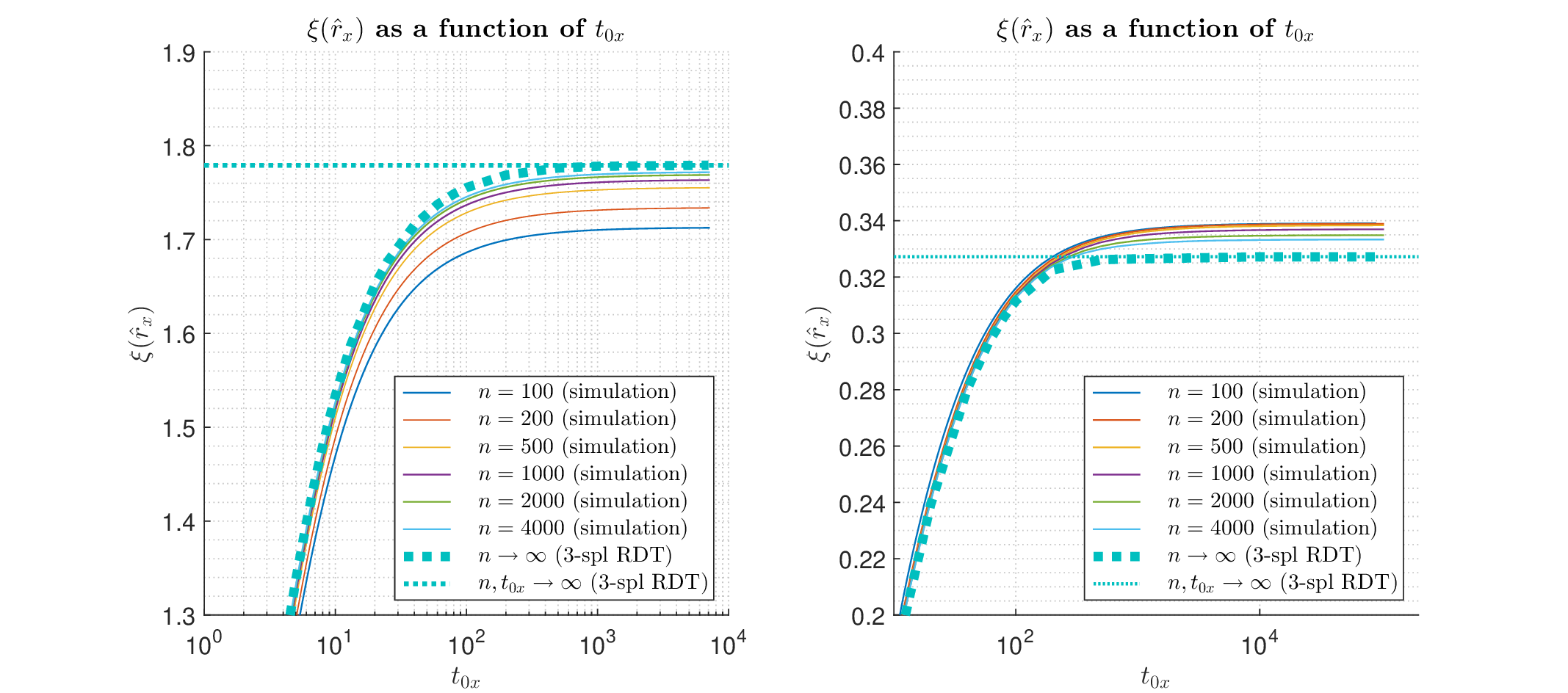}}
\caption{Convergence of $\bar{\xi}(\hat{r}_x)$ as $n$ grows;  $+$Hop (left) and $-$Hop (right) }
\label{fig:figconv1}
\end{figure}
\begin{figure}[h]
\centering
\centerline{\includegraphics[width=1.00\linewidth]{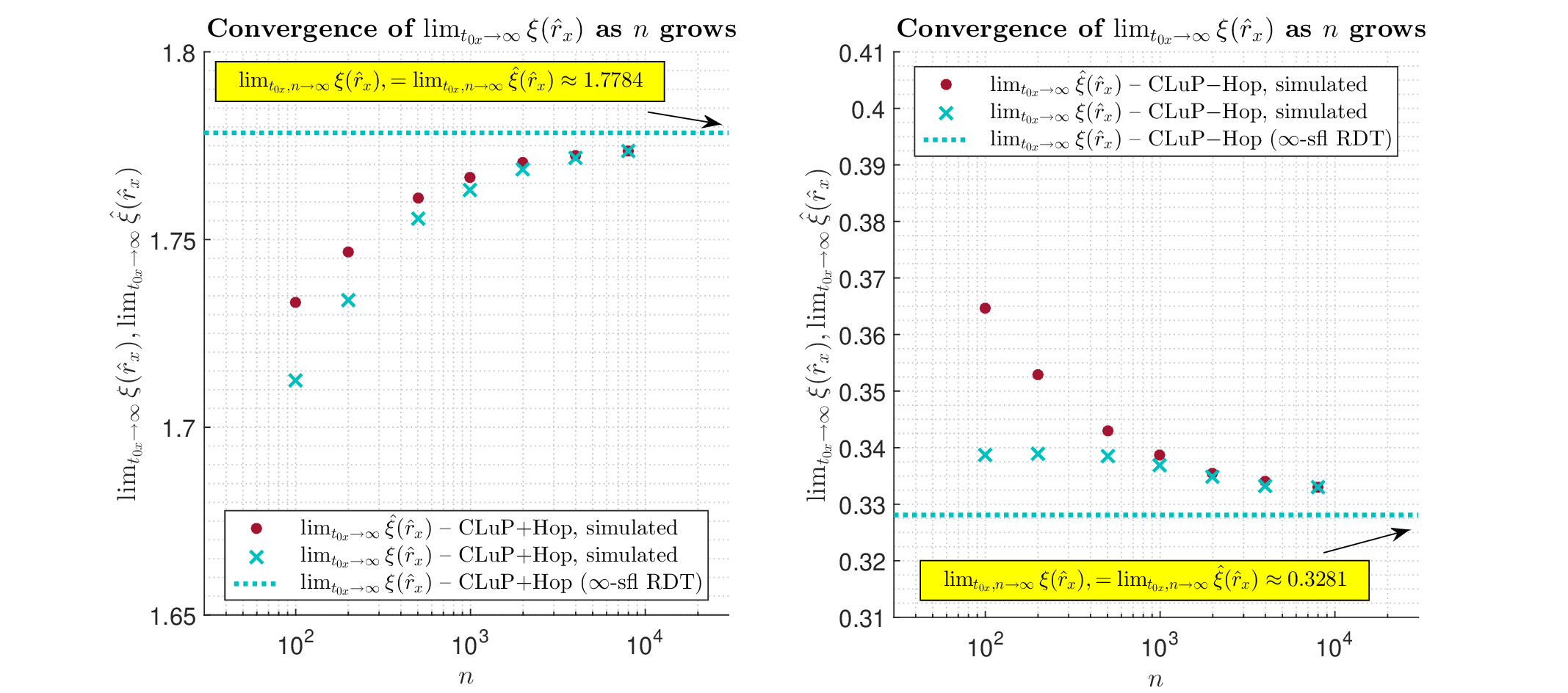}}
\caption{Convergence of $\lim_{t_{0x}\rightarrow\infty} \xi(\hat{r}_x)$ and $\lim_{t_{0x}\rightarrow\infty} \hat{\xi}(\hat{r}_x)$ as $n$ grows;  $+$Hop (left) and $-$Hop (right) }
\label{fig:figconv2}
\end{figure}

\noindent \underline{\textbf{\emph{Concentrations}}}

Closely related to the above convergence are the concentration properties. Theory predicts that for $n\rightarrow\infty$ all key underlying quantities rapidly concentrate. Figure \ref{fig:figconc} shows a finite dimensional emulation of concentration effects. In the upper portion of the figure we have $+$Hop and in the lower portion $-$Hop concentration results. For both scenarios we have taken $n=200$ and $m=4000$ as representatives of smaller and larger $n$. One observes a rapid onset of concentrations effects with standard deviations dropping $7-8$ times for the given $n$ span. Moreover, we superimposed Gaussian densities with means and standard deviations that match the ones obtained through  simulations. As figure shows, they remarkably well fit the simulated histograms.

\begin{figure}[h]
\begin{minipage}[b]{1\linewidth}
\centering
\centerline{\includegraphics[height=2.7in,width=1\linewidth]{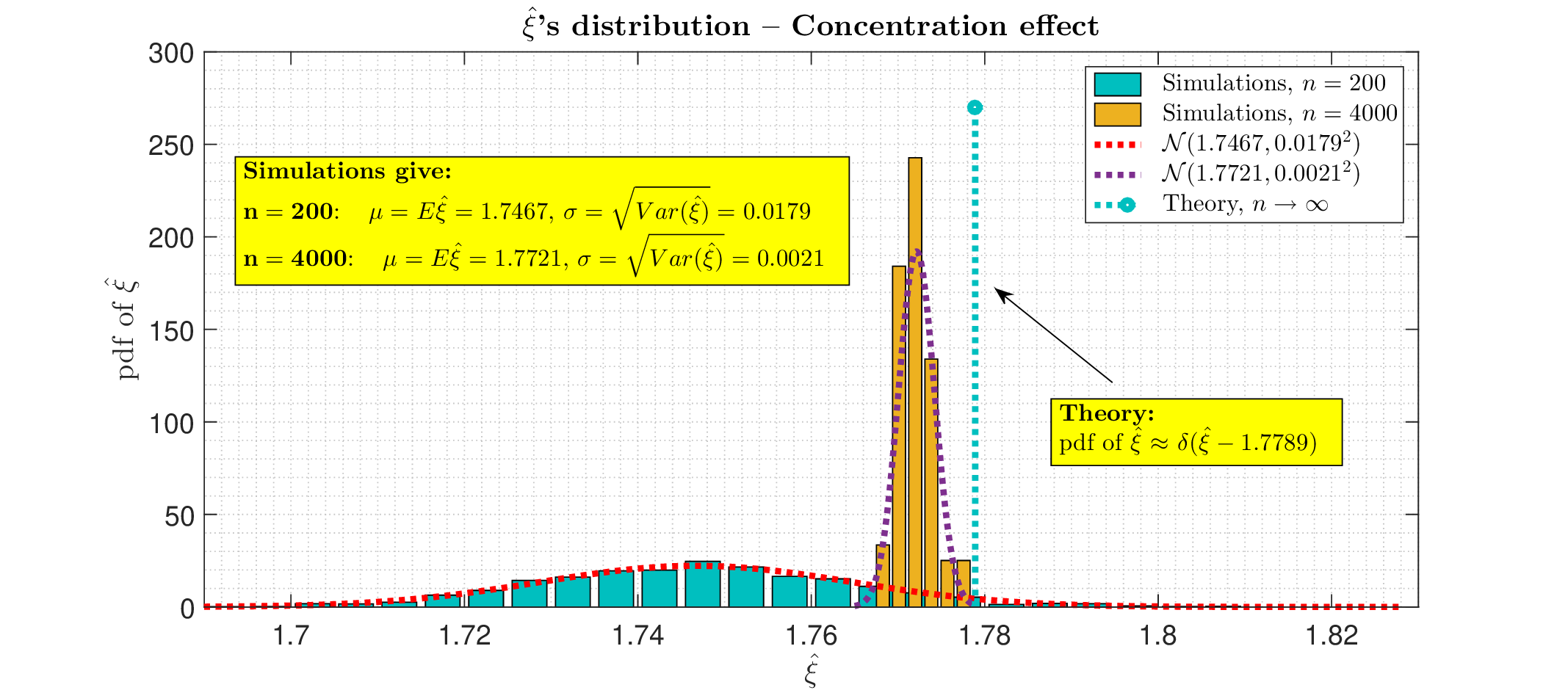}}
\end{minipage}
\begin{minipage}[b]{1\linewidth}
\centering
\centerline{\includegraphics[height=2.7in,width=1\linewidth]{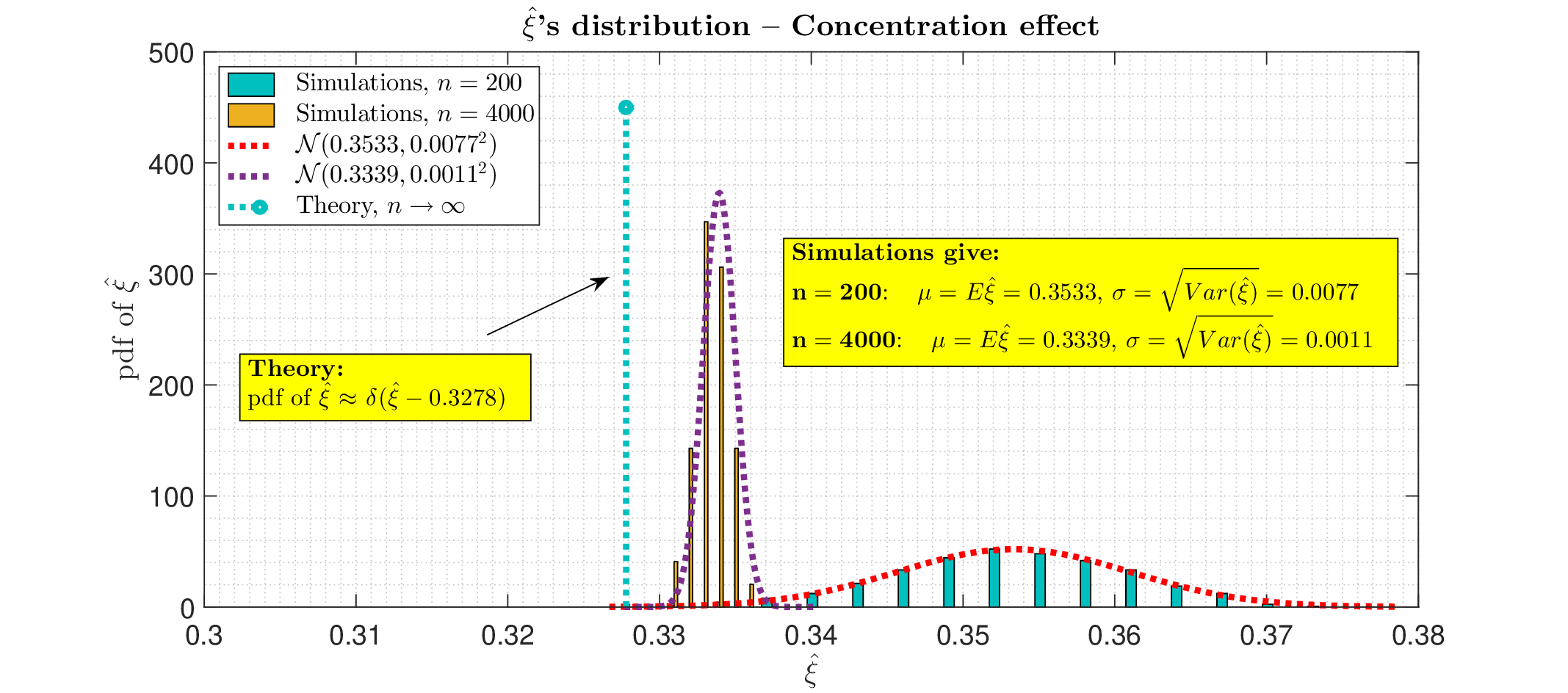}}
\end{minipage}
\caption{Concentration effect;  $+$Hop (top) and $-$Hop (bottom)}
\label{fig:figconc}
\end{figure}

\noindent \underline{\textbf{\emph{Objective shape as a function of $r_x$}}}

In Figure \ref{fig:fig9a}, we show the scaled objective $\frac{\bar{f}_{b} \lp r_x \rp }{t_{0x}}$ for a concrete  $t_{0x} =20$.  For both $+$Hop and $-$Hop we observe that $\frac{\bar{f}_{b} \lp r_x \rp }{t_{0x}}$  has no local optima. We also repeated the very same test across a wide range of $t_{0x}$ and observed the same trend. Such an absence of objective's local optima is a necessary condition for successful generic running of descending algorithms. To what degree (if any) intrinsic features other than objective landscape impact performance of descending algorithms  remains to be seen. Studying potential presence/absence of ``near optimal'' solutions clustering organization might be interesting  next steps in these directions. In particular, overlap gap properties (OGP) \cite{Gamar21,GamarSud14,GamarSud17,GamarSud17a,AchlioptasCR11,HMMZ08,MMZ05} and local entropies (LE) \cite{Bald15,Bald16,Bald20} are two clustering related concepts predominantly present in recent literature.  We should also add that the shape of the objective critically depends on the accuracy of the underlying numerical evaluations. As mentioned earlier, all our evaluations are done on the third partial level of lifting. Tiny corrections do appear on higher lifting levels but remain visually almost undetectable and without significant impact on objective's shape.

\begin{figure}[h]
\centering
\centerline{\includegraphics[width=1\linewidth]{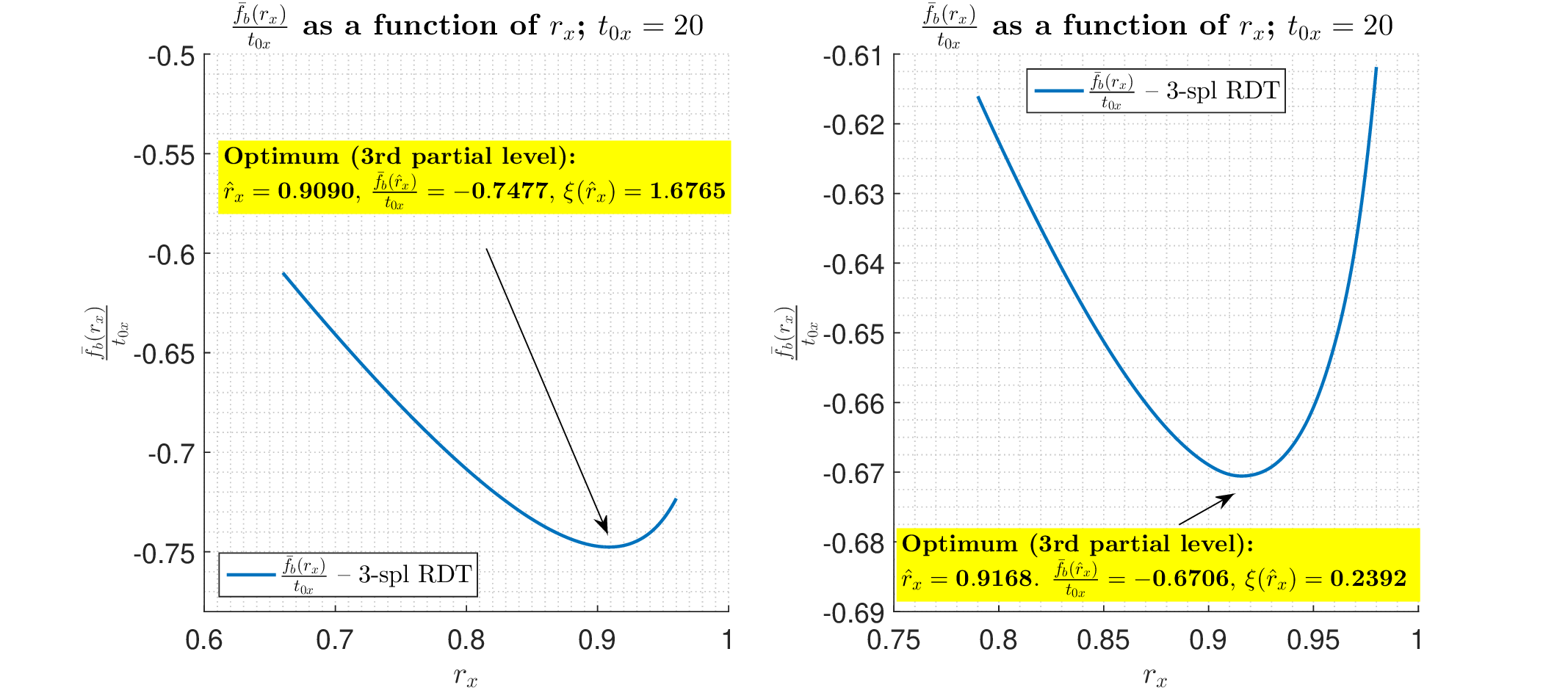}}
\caption{Landscape as a function of $r_x$;  $+$Hop (left) and $-$Hop (right)}
\label{fig:fig9a}
\end{figure}

\noindent \underline{\textbf{\emph{Restarting/retuning}}}

As mentioned earlier, all our results are obtained for basic $\pm$Hop implementations. Even though the obtained results are well above the best of hopes, significant improvement in certain scenarios can be achieved with fairly minimal modifications/adaptations. For example, just restarting with different randomly chosen initial configurations already  helps a lot. Retuning $c^{(t)}$ as the algorithm progresses provides an additional help as well. To give a flavor as to what kind of improvement one can expect, we in Table \ref{tab:tabrestart} show the impact of restarting+retuning modification on $-$Hop (in our experience $-$Hop seems to benefit the most from additional modifications). We limited number of restarts to 100 (even though 20 was typically sufficient) and retuned $c^{(t)}$ as $c^{(t)}=\mbox{Unif}[1,1.3]$. We selected two scenarios $n=500$ and $n=2000$ to emphasize that a strong improvement  is possible across a range of $n$ including those on the order of few thousands where the theoretical limits are already being approached fairly closely.

\begin{table}[h]
\caption{Restarting/retuning effects on CLuP$-$Hop algorithm -- \bl{\textbf{plain}}/\prp{\textbf{restart+retune}}  }\vspace{.1in}
\centering
\def\arraystretch{1.2}
\begin{tabular}{||l||c||c||c||c|| }\hline\hline
 \hspace{-0in}$n$                                             & $500$    & $2000$ &   $\infty$ (\textbf{theory}) \\ \hline\hline
$\hat{\xi}$ ($-$Hop)                                         & \bl{$\mathbf{0.3430}$}  & \bl{$\mathbf{0.3355}$} & $\mathbf{0.3281}$  \\ \hline\hline
$\hat{\xi}$ ($-$Hop)                                         & \prp{$\mathbf{0.3358}$}  & \prp{$\mathbf{0.3330}$} & $\mathbf{0.3281}$  \\ \hline\hline
\end{tabular}
\label{tab:tabrestart}
\end{table}

\noindent \underline{\textbf{\emph{Overlaps}}}

Excellent CLuP$\pm$Hop performance  allows simulation of near optimal solutions -- configurations that produce energies in the vicinity of the optimal ground state one. Of particular interest are the overlaps structures together with associated GIbbs measures. In the thermodynamic limit the key components of the Gibbs measures concentrate on $\p_2,\p_3,\dots,\p_r$ and $\q_2,\q_3,\dots,\q_r$. In Figure \ref{fig:fig8} and \ref{fig:fig9}  we show how the structure of $\p$ and $\q$ changes as lifting progresses. Visual presentation is somewhat facilitated if one associates with $\p$, $\q$, and $\c$ the following
\begin{eqnarray}\label{eq:ultmet1}
\bl{\mbox{\textbf{$\p\lp \frac{\c}{\c_2}\rp $ map:}}}  \hspace{.2in} & & \p_2 \leftrightarrow   \left [ \frac{\c_3}{\c_2}, \frac{\c_2}{\c_2} \right ],\quad
\p_3  \leftrightarrow  \left [ \frac{\c_4}{\c_2}, \frac{\c_3}{\c_2}  \right ], \quad \dots,
\nonumber \\
\bl{\mbox{\textbf{$\q\lp \frac{\c}{\c_2}\rp $ map:}}}  \hspace{.2in} & & \q_2 \leftrightarrow   \left [ \frac{\c_3}{\c_2}, \frac{\c_2}{\c_2} \right ],\quad
\q_3  \leftrightarrow  \left [ \frac{\c_4}{\c_2}, \frac{\c_3}{\c_2}  \right ], \quad \dots,
\end{eqnarray}
for $+$Hop model and
\begin{eqnarray}\label{eq:ultmet1a0}
\bl{\mbox{\textbf{$\p \lp \frac{\c}{\c_{\infty}}\rp $ map:}}}  \hspace{.2in}   & & \p_2 \leftrightarrow   \left [ \frac{\c_2}{\c_{\infty}}, \frac{\c_3}{\c_{\infty}} \right ],\quad
\p_3  \leftrightarrow  \left [ \frac{\c_3}{\c_{\infty}}, \frac{\c_4}{\c_{\infty}}  \right ], \quad \dots, \nonumber \\
\bl{\mbox{\textbf{$\q\lp \frac{\c}{\c_{\infty}}\rp $ map:}}}  \hspace{.2in}  & &  \q_2 \leftrightarrow   \left [ \frac{\c_2}{\c_{\infty}}, \frac{\c_3}{\c_{\infty}} \right ],\quad
\q_3  \leftrightarrow  \left [ \frac{\c_3}{\c_{\infty}}, \frac{\c_4}{\c_{\infty}}  \right ], \quad \dots.
\end{eqnarray}
for $-$Hop model. Concrete numerical values for $\p$, $\q$, and $\c$ up to the 6th lifting level are given in Tables \ref{tab:2rsbunifiedsqrtpos} and \ref{tab:NEG2rsbunifiedsqrtpos} \cite{Stojnichopflrdt23} (Table  \ref{tab:2rsbunifiedsqrtpos} relates to $+$Hop and Table \ref{tab:NEG2rsbunifiedsqrtpos}  to $+$Hop model; in addition to $\p$, $\q$, and $\c$ the $r$th level values of $f_{sq}^{+}(\infty)$ and $f_{sq}^{-}(\infty)$, $f_{sq}^{+,r}(\infty)$ and $f_{sq}^{-,r}(\infty)$ , are given as well; as stated above, results from the tables are also visualized in
Figures \ref{fig:fig8} and \ref{fig:fig9}). Moreover,  $\p(\cdot)$ and $\q(\cdot)$ maps are complemented with simulated near optimal configurations overlaps  (we ran randomly restarted algorithm's variant). The simulated distributions fairly closely match the $\p$ and $\q$ Gibbs measures cdfs. Also, comparing Figures \ref{fig:fig8} and \ref{fig:fig9} one observes a fundamentally different behavior of $+$Hop and $-$Hop  overlaps. Both $\p$ and $\q$ $+$Hop overlaps are typically close to one which means that near optimal solutions $\x$ (as well as their corresponding $\y$'s) are ``\emph{typically close}'' to each other. On the other hand,  in $-$Hop model they are ``\emph{typically far away}'' and usually almost orthogonal to each other. We should add that due to the dependence on the choice of near optimality (we have taken $\pm.003$ as allowed deviation from the algorithm's best solution) simulated results are more an indication than accurate value (simulating exact values seems as a conceptually rather hard task).
\begin{table}[h]
\caption{$r$-sfl RDT parameters; $+$Hop model; $\alpha=1$; $\hat{\c}_1\rightarrow 1$; $n,\beta\rightarrow\infty$}\vspace{.1in}
\centering
\def\arraystretch{1.2}
\begin{tabular}{||l||c||c||c||c||c||}\hline\hline
 \hspace{-0in}$r$                                             & $\hat{\gamma}_{sq}$    & $\hspace{-.05in}\begin{bmatrix}\hat{\p}_{r-1}  & \hat{\p}_{r-2} & \dots & \hat{\p}_{1}    \end{bmatrix}^{T^{\big.}}_{\big.} \hspace{-.05in}$        & $\hspace{-.05in}\begin{bmatrix}\hat{\q}_{r-1}  & \hat{\q}_{r-2} & \dots & \hat{\q}_{1}    \end{bmatrix}^{T^{\big.}}_{\big.} \hspace{-.05in}$ &  $\hspace{-.05in}\begin{bmatrix}\hat{\c}_{r}  & \hat{\c}_{r-1} & \dots & \hat{\c}_{2}    \end{bmatrix}^{T^{\big.}}_{\big.} \hspace{-.05in}$    & $f_{sq}^{+,r}$  \\ \hline\hline
  $\mathbf{2}$                                       & $0.6173$ & $\begin{bmatrix} \rightarrow 1\end{bmatrix}_{\big.}^{\big.}$  & $\begin{bmatrix}  \rightarrow 1\end{bmatrix}_{\big.}^{\big.}$  &  $ \begin{bmatrix}0.4246\end{bmatrix}_{\big.}^{\big.}$   & \bl{$\mathbf{1.7832}$} \\ \hline
  $\mathbf{3}$                                       & $0.7434$ & $\begin{bmatrix} 0.7510  \\ \rightarrow 1 \end{bmatrix}_{\big.}^{\big.}$ &  $\begin{bmatrix} 0.8397 \\ \rightarrow 1\end{bmatrix}$
   & $\begin{bmatrix} 0.2508 \\ 1.7762  \end{bmatrix}$ & \bl{$\mathbf{1.7791}$}  \\ \hline
  $\mathbf{4}$                                       & $0.8024$ & $\begin{bmatrix} 0.5900 \\ 0.9097 \\ \rightarrow 1 \end{bmatrix}_{\big.}^{\big.}$ & $\begin{bmatrix} 0.6927 \\ 0.9546 \\ \rightarrow 1\end{bmatrix}$
   &  $\begin{bmatrix} 0.1829 \\ 0.8349 \\3.9792 \end{bmatrix}$   & \bl{$\mathbf{1.77859}$}  \\ \hline
  $\mathbf{5}$                                       & $0.8203$ & $\begin{bmatrix} 0.4230 \\ 0.97700\\  0.9473 \\ \rightarrow 1 \end{bmatrix}_{\big.}^{\big.}$ & $\begin{bmatrix} 0.5205 \\ 0.8556 \\0.9762 \\ \rightarrow 1\end{bmatrix}$
   &  $\begin{bmatrix} 0.1214 \\ 0.4705 \\ 1.4078 \\5.6138 \end{bmatrix}$   & \bl{$\mathbf{1.77846}$}  \\ \hline
  $\mathbf{6}$                                       & $0.8235$ & $\begin{bmatrix} 0.3030 \\ 0.6132 \\ 0.8398\\  0.9602\\ \rightarrow 1 \end{bmatrix}_{\big.}^{\big.}$ & $\begin{bmatrix} 0.3786\\ 0.7141 \\ 0.9093 \\0.9833 \\ \rightarrow 1\end{bmatrix}$
   &  $\begin{bmatrix} 0.0883 \\ 0.3041 \\ 0.7092 \\ 1.8885 \\6.3850 \end{bmatrix}$   & \bl{$\mathbf{1.77842}$}  \\ \hline \hline
\end{tabular}
\label{tab:2rsbunifiedsqrtpos}
\end{table}

\begin{table}[h]
\caption{$r$-sfl RDT parameters; $-$Hop model; $\alpha=1$; $\hat{\c}_1\rightarrow 1$; $n,\beta\rightarrow\infty$}\vspace{.1in}
\centering
\def\arraystretch{1.2}
\begin{tabular}{||l||c||c||c||c||c||}\hline\hline
 \hspace{-0in}$r$                                             & $\hat{\gamma}_{sq}$    & $\hspace{-.05in}\begin{bmatrix}\hat{\p}_{r-1}  & \hat{\p}_{r-2} & \dots & \hat{\p}_{1}    \end{bmatrix}^{T^{\big.}}_{\big.} \hspace{-.05in}$        & $\hspace{-.05in}\begin{bmatrix}\hat{\q}_{r-1}  & \hat{\q}_{r-2} & \dots & \hat{\q}_{1}    \end{bmatrix}^{T^{\big.}}_{\big.} \hspace{-.05in}$ &  $\hspace{-.05in}\begin{bmatrix}\hat{\c}_{r}  & \hat{\c}_{r-1} & \dots & \hat{\c}_{2}    \end{bmatrix}^{T^{\big.}}_{\big.} \hspace{-.05in}$    & $f_{sq}^{-,r}$  \\ \hline\hline
  $\mathbf{2}$                                       & $0.1654$ & $\begin{bmatrix} \rightarrow 1\end{bmatrix}_{\big.}^{\big.}$  & $\begin{bmatrix}  \rightarrow 1\end{bmatrix}_{\big.}^{\big.}$  &  $ \begin{bmatrix}2.6916\end{bmatrix}_{\big.}^{\big.}$   & \bl{$\mathbf{0.3202}$} \\ \hline
  $\mathbf{3}$                                       & $0.1748$ & $\begin{bmatrix} 0.5722  \\ \rightarrow 1 \end{bmatrix}_{\big.}^{\big.}$ &  $\begin{bmatrix} 0.0599 \\ \rightarrow 1\end{bmatrix}$
   & $\begin{bmatrix} 10.548 \\ 2.2264  \end{bmatrix}$ & \bl{$\mathbf{0.3272}$}  \\ \hline
  $\mathbf{4}$                                       & $0.1766$ & $\begin{bmatrix} 0.3639 \\ 0.6836 \\ \rightarrow 1 \end{bmatrix}_{\big.}^{\big.}$ & $\begin{bmatrix} 0.0107 \\ 0.1282 \\ \rightarrow 1\end{bmatrix}$
   &  $\begin{bmatrix} 27.9800 \\ 5.0700 \\ 2.1306 \end{bmatrix}$   & \bl{$\mathbf{0.3279}$}  \\ \hline
  $\mathbf{5}$                                       & $0.1776$ & $\begin{bmatrix} 0.2640 \\ 0.5137 \\  0.7307 \\ \rightarrow 1 \end{bmatrix}_{\big.}^{\big.}$ & $\begin{bmatrix} 0.0035 \\ 0.0387 \\0.1742 \\ \rightarrow 1\end{bmatrix}$
   &  $\begin{bmatrix} 54.3693 \\ 9.3081 \\ 3.7962 \\ 2.1046 \end{bmatrix}$   & \bl{$\mathbf{0.32803}$}  \\ \hline
  $\mathbf{6}$                                       & $0.1774$ & $\begin{bmatrix} 0.2040 \\ 0.3982 \\ 0.5835\\  0.7556\\ \rightarrow 1 \end{bmatrix}_{\big.}^{\big.}$ & $\begin{bmatrix} 0.0016\\ 0.0148 \\ 0.0648 \\0.2036 \\ \rightarrow 1\end{bmatrix}$
   &  $\begin{bmatrix} 85.9067 \\ 15.8758 \\ 6.2406 \\ 3.2936 \\ 2.0872 \end{bmatrix}$   & \bl{$\mathbf{0.32807}$}  \\ \hline \hline
\end{tabular}
\label{tab:NEG2rsbunifiedsqrtpos}
\end{table}

\begin{figure}[h]
\centering
\centerline{\includegraphics[width=1.00\linewidth]{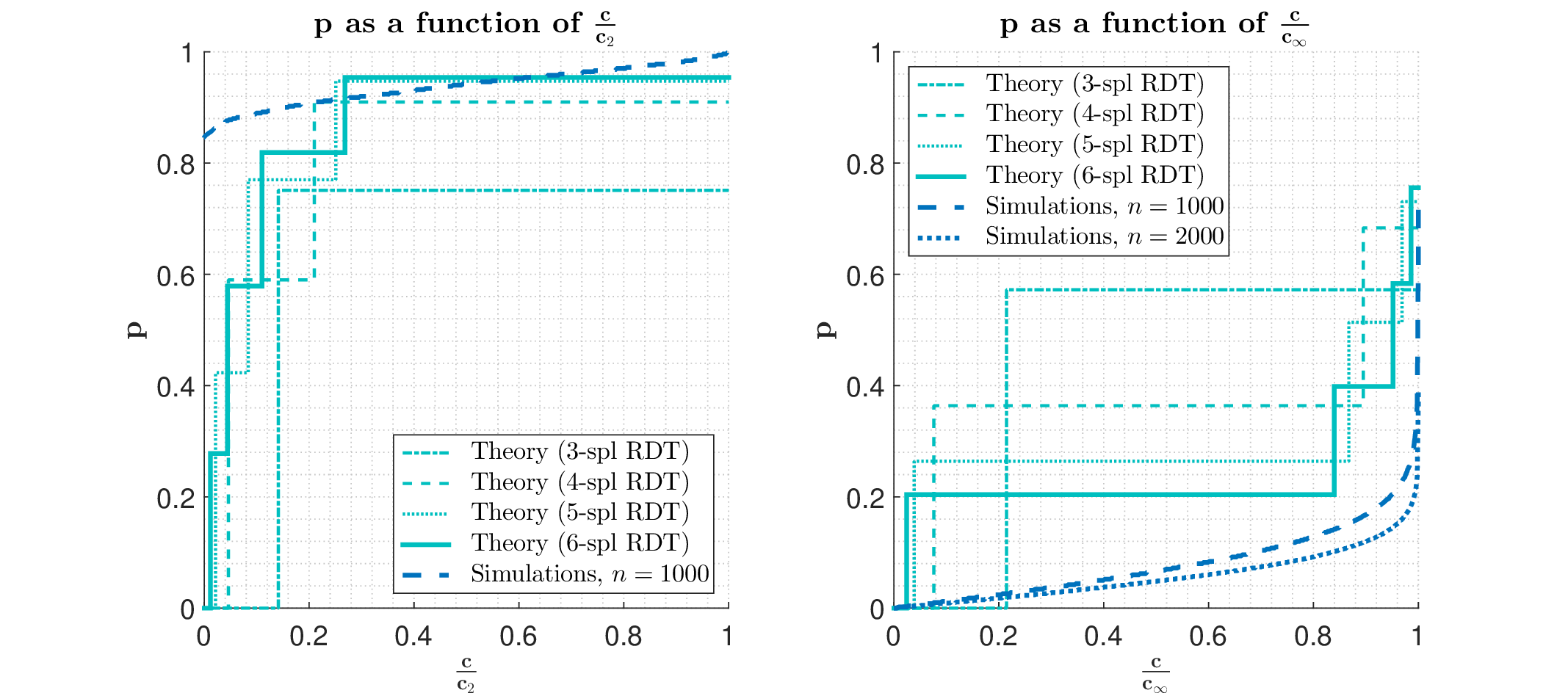} }
\caption{$\p$ overlaps;  $+$Hop (left) and $-$Hop (right)}
\label{fig:fig8}
\end{figure}

\begin{figure}[h]
\centering
\centerline{\includegraphics[width=1.00\linewidth]{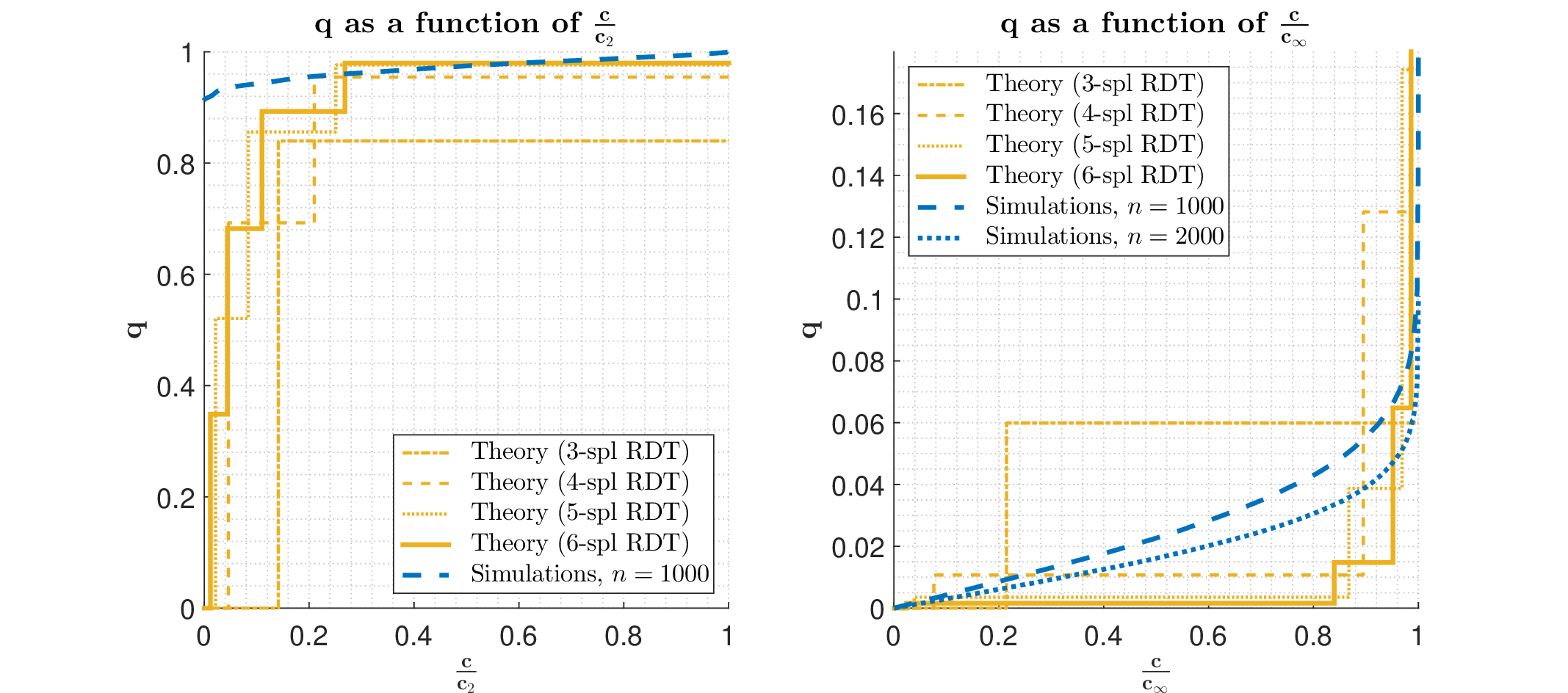} }
\caption{$\q$ overlaps;  $+$Hop (left) and $-$Hop (right)}
\label{fig:fig9}
\end{figure}

We complement $\pm$Hop models overlaps results with the corresponding SK ones. In Table  \ref{tab:7rsbSK} we show the SK lifting progression up to the 7th level. Results from Table  \ref{tab:7rsbSK}  are visualized in Figure \ref{fig:fig10} (since there is a rather large number of numerical evaluations increasing their level of precision might slightly change the 6th and 7th level numbers in Table  \ref{tab:7rsbSK}; however, shape and location of the corresponding curves in Figure \ref{fig:fig10} as well as $f_{csk}^{(r)}$ (the $r$th level SK ground state free energy) are not expected to significantly change). It is interesting to note that SK overlaps cdf is much more similar to $+$Hop than to $-$Hop. Furthermore, the SK overlap cdf seems to be higher than the simple approximation given in \cite{OppShe05} (denoted as $\infty$-RSB estimate in Figure  \ref{fig:fig10}). On the other hand, we obtain $f_{csk}^{(7)}(\infty)\approx 0.76319$ as the SK model ground state free energy on the 7th lifting level. This indicates that $0.76321\pm 0.00003$  prediction of  \cite{CrisRizo02} and $\approx0.76317$ prediction of    \cite{OppSch08,OppSS07} are indeed close to the true value. The simulated SK overlap values are shown as well  (we ran restarted algorithm's variant together with random deviations in $t_{0x}$). One should note though, that (compared to $\pm$Hop) overlaps simulations in SK case seem more sensitive to the choice of near optimality and as such should be taken even more cautiously as an indication of the true behavior. Nonetheless, the trend in the bulk of the distribution  is very much in alignment with the theoretical prediction.

\begin{table}[h]
\caption{$r$-sfl RDT parameters; SK model;  $\hat{\q}_1,\hat{\c}_1\rightarrow 1$; $n,\beta\rightarrow\infty$}\vspace{.1in}
\centering
\def\arraystretch{1.2}
{\footnotesize
\begin{tabular}{||l||c|c|c|c|c||c|c|c|c|c|c||c||}\hline\hline
 \hspace{-0in}$r$                                              & $\hat{\q}_5$     & $\hat{\q}_4$   &   $\hat{\q}_4$  & $\hat{\q}_3$  & $\hat{\q}_2$  & $\hat{\c}_5$ & $\hat{\c}_4$ &  $\hat{\c}_5$   &  $\hat{\c}_4$   &   $\hat{\c}_3$   &   $\hat{\c}_2$    & $f_{csk}^{(r)}(\infty)$  \\ \hline\hline
 $\mathbf{2}$                                       & $0$  & $0$ & $0$ & $0$ & $0$ & $\rightarrow 0$ & $\rightarrow 0$ &  $\rightarrow 0$
 &  $\rightarrow 0$ &  $\rightarrow 0$  &  $   0.5779 $  & \bl{$\mathbf{0.76883}$} \\ \hline
   $\mathbf{3}$                                      & $0$  & $0$ & $0$  & $0$ & $0.7434$ &  $\rightarrow 0$  & $\rightarrow 0$ &  $\rightarrow 0$
 &  $\rightarrow 0$ &  $0.3569$
 &  $1.4586$   & \bl{$\mathbf{0.76403}$}   \\ \hline
  $\mathbf{4}$                                       & $0$ & $0$ & $0$  & $ 0.5587$ & $ 0.9088$ & $\rightarrow 0$  & $\rightarrow 0$ &  $\rightarrow 0$
 &  $0.2599$ &  $  0.8280$
 &  $ 2.5103$   & \bl{$\mathbf{0.76341}$}   \\ \hline
 $\mathbf{5}$                                      & $0$ & $0$ & $  0.4421$  & $ 0.7912 $ & $ 0.9588$ & $\rightarrow 0$ & $\rightarrow 0$  & $0.2048 $  &  $0.6104$ &  $1.3026$
 &  $ 3.7571$   & \bl{$\mathbf{0.76326}$}   \\ \hline
  $\mathbf{6}$                                      & $0$ & $0.3708  $ & $  0.6952  $  & $ 0.8939 $ & $ 0.9796 $ & $\rightarrow 0$ & $0.1718$   & $0.5027  $  &  $0.9485 $   &  $1.8802  $
 &  $ 5.3556 $   & \bl{$\mathbf{0.76321}$}   \\ \hline
  $\mathbf{7}$                                      &  $ \hspace{-.03in} 0.2854 \hspace{-.03in} $  & $0.5363  $ & $  0.7621  $  & $ 0.9133 $ & $ 0.9829  $ & $ 0.1381   $ & $0.3752 $   & $0.6383  $  &  $1.08235 $   &  $2.0709  $
 &  $ 5.8472 $   & \bl{$\mathbf{0.76319}$}   \\ \hline
 \hline
\end{tabular}
}
\label{tab:7rsbSK}
\end{table}

\begin{figure}[h]
\centering
\centerline{\includegraphics[width=1.00\linewidth]{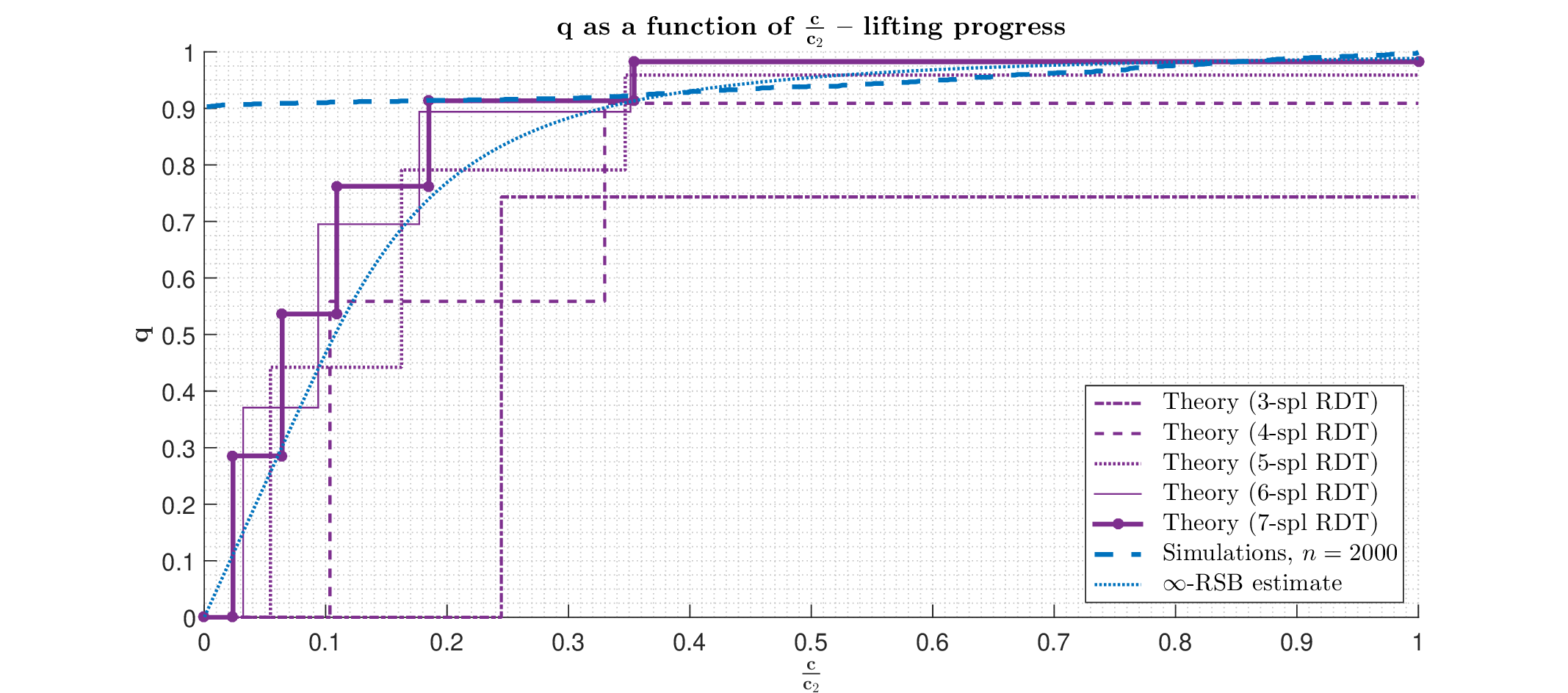} }
\caption{$\q$ overlap;  SK model}
\label{fig:fig10}
\end{figure}

\section{Conclusion}
\label{sec:conc}

We studied algorithmic aspects of \emph{positive} and \emph{negative} Hopfield ($\pm$Hop) models and focused on determining their ground state free energies. These problems are equivalent to binary  maximization of random positive/negative semi-definite quadratic forms. Their indefinite quadratic form analogue is the celebrated SK model. Following the success of \emph{Controlled Loosening-up} (CLuP-SK) algorithms in computing near ground state free energies of SK models \cite{Stojnicclupsk25}, we here proposed CLuP$\pm$Hop analogues for $\pm$Hop models. An excellent performance is observed already  for $n$  as small as few thousands. In particular, we have that: (i)  CLuP+Hop achieves $\sim 1.77$ as the ground state free energy of the positive model; and (ii)  CLuP-Hop achieves $\sim 0.33$ as the corresponding energy of the negative model. Both of these closely approach the theoretical thermodynamic ($n\rightarrow\infty$) limits $\approx 1.7784$ and $\approx 0.3281$ and position computation of near ground state free energies of  $\pm$Hop models as \emph{typically} easy problems.

To analyze the introduced algorithms we associated the CLuP$\pm$Hop models and utilized Fully lifted random duality theory (fl RDT) \cite{Stojnicflrdt23} to study them. A generic analytical framework was developed and employed to characterize the entire dynamics of the proposed algorithms. Already on the third (partial) level of lifting (3-spf RDT) we obtained almost identical match between the theoretical predictions and algorithmically simulated results. Excellent algorithmic performance and generality of the analytical concepts allowed to uncover a host of interesting features as well. Strong convergence and concentration properties are observed (already  for $n\sim 1000$) and a favorable -- no local optima -- landscape of the underlying algorithmic objective is uncovered.

Lifting up to the 3rd level is usually sufficient to obtain very precise characterizations of almost any associated performance measure. One notable exception is the characterization of the configurational overlaps and their GIbbs measures. To handle these much higher levels of lifting are needed. We conducted evaluations up to the 6th lifting level (6-spl RDT) and along the way uncovered a remarkable fundamental intrinsic difference between $+$Hop and $-$Hop models. Typical near optimal configurations are \emph{close} to each other for $+$Hop and \emph{far away} for $-$Hop model.

The introduced concepts  are very generic and allow for many extensions as well. Developing CLuP analogues for various models discussed in \cite{Stojnictcmspnncapliftedrdt23,Stojnicnflgscompyx23,Stojnicsflgscompyx23,Stojnicflrdt23} forms just a small fraction of practically endless possibilities. In addition to relying on the main concepts presented here and in \cite{Stojnicclupsk25,Stojnicclupint19,Stojnicclupspreg20}, these extensions also require a bit of technical adjustments that are usually problem specific. We discuss these in separate papers.

\begin{singlespace}
\bibliographystyle{plain}
\bibliography{nflgscompyxRefs}
\end{singlespace}

\end{document}